\documentclass[10pt,onecolumn]{IEEEtran} 

\usepackage{times}
\usepackage{amsmath,amssymb,mathrsfs,amsthm}
\usepackage{graphicx}
\usepackage{color}
\usepackage{subfig}
\usepackage{cite,url}
\usepackage{multirow,tabularx,makecell}
\usepackage[ruled,lined]{algorithm2e}
\usepackage{algorithmicx}
\usepackage[compatible]{algpseudocode} 
\usepackage{bbm}
\usepackage{wasysym}
\usepackage{stackrel}
\usepackage{tikz}
\usetikzlibrary{matrix,arrows,decorations.pathreplacing,automata,positioning,decorations.pathmorphing}
\usepackage{float}
\usetikzlibrary{shapes.multipart}

\newtheorem{theorem}{{\bf Theorem}}
\newtheorem{lemma}{{\bf Lemma}}

\newtheorem{corollary}{{\bf Corollary}}

\newtheorem{example}{{\bf Example}}
\newtheorem{definition}{\bf Definition}

\newtheorem{remark}{Remark}

\usepackage{hyperref}

\newcommand{\net}{\mathcal{N}}
\newcommand{\netnodes}{V}
\newcommand{\netedges}{E}

\newcommand{\compgraph}{\mathcal{G}}
\newcommand{\compnodes}{\Omega}
\newcommand{\compedges}{\Gamma}

\newcommand{\preedges}[1]{\Phi_{\uparrow}(#1)}
\newcommand{\sucedges}[1]{\Phi_{\downarrow}(#1)}

\newcommand{\prefunction}[1]{\Lambda_{\uparrow}(#1)}
\newcommand{\sucfunction}[1]{\Lambda_{\downarrow}(#1)}
\newcommand{\head}[1]{\mathsf{head}(#1)}
\newcommand{\tail}[1]{\mathsf{tail}(#1)}

\newcommand{\embedding}{\mathcal{E}}
\newcommand{\setofembeddings}{\mathbb{E}}
\newcommand{\rembedding}{\textit{R-Embedding}}
\newcommand{\setofpaths}{\Sigma}
\newcommand{\pathstart}[1]{\mathsf{start}(#1)}
\newcommand{\pathend}[1]{\mathsf{end}(#1)}

\newcommand{\maxcut}{\textit{SIMPLE MAX CUT}}
\newcommand{\mincut}{\textit{2-Cut}}

\newcommand{\mincostold}{\textsf{MinCost}$(\complement)$}
\newcommand{\mincostnew}{\textsf{MinCost}$(C)$}
\DeclareMathOperator*{\argmin}{arg\,min}

\begin{document}

\title{On the Maximum Rate of Networked Computation in a Capacitated Network}

\author{ \authorblockN{Pooja Vyavahare\authorrefmark{1}, Nutan
    Limaye\authorrefmark{2}, Ajit A. Diwan\authorrefmark{2}, D. Manjunath\authorrefmark{1}
    } \\
     \authorblockA{\authorrefmark{1}  Department of Electrical
    Engineering, IIT Bombay \\ 
    \{vpooja,dmanju\}@ee.iitb.ac.in} \\
    \authorblockA{\authorrefmark{2}  Department of Computer Science and 
    Engineering, IIT Bombay \\ 
    \{nutan,aad\}@cse.iitb.ac.in}
    \thanks{Pooja Vyavahare and D. Manjunath are affiliated with the Bharti Center for
Communications. Their work has been partially supported by grants from DST
and CEFIPRA. Pooja Vyavahare also received support from ITRA. Nutan Limaye is supported by grants from DST, DAAD and CEFIPRA.}

}
\maketitle
\begin{abstract}
Given a capacitated communication network $\net$ and a function $f$ that needs to be computed on $\net,$ we study the problem of generating a computation and communication schedule in $\net$ to maximize the rate of computation of $f$. Shah et. al.[IEEE Journal of Selected Areas in Communication, 2013] studied this problem when the computation schema $\compgraph$ for $f$ is a tree graph. We define the notion of a schedule when $\compgraph$ is a general DAG and show that finding an optimal schedule is equivalent to
finding the solution of a packing linear program. 
 
We prove that approximating the maximum rate is MAX SNP-hard by looking at the packing LP. For this packing LP we prove that solving the separation oracle of its dual is equivalent to solving the LP. The separation oracle of the dual reduces to the problem of finding \emph{minimum cost embedding} given $\net,\compgraph,$ which we prove to be MAX SNP-hard even when $\compgraph$ has bounded degree and bounded edge weights and $\net$ has just three vertices. We present a polynomial time algorithm to compute the maximum rate of function computation when $\net$ has two vertices by reducing the problem to a version of submodular function minimization problem.

For the general $\net$ we study restricted class of schedules and its equivalent packing LP. We observe that for this packing LP also the separation oracle of its dual reduces to finding minimum cost embedding. A version of this minimum cost embedding problem has been studied in literature and we relate our cost model with the one present in literature. We present a quadratic integer program for the minimum cost embedding problem and its linear programming relaxation based on earthmover  metric. Applying the randomized rounding techniques to the optimal solution of this LP we give approximate algorithms for some special class of graphs. We present constant factor approximation algorithms for maximum rate when $\compgraph$ is a bounded width layered graph and when it is a planar graph with bounded out-degree. We also present $O(D\log n)$-approximation algorithm for arbitrary DAG $\compgraph$ where $D$ is the maximum out-degree of a vertex in $\compgraph$ and $n$ is the number of vertices in $\net.$ We also prove that if a DAG has a spanning tree in which every edge is a part of $O(F)$ fundamental cycles then there is a $O(FD)$-approximation algorithm.

\begin{keywords}
In-network computation, 
maximum computation rate, 
minimum cost of computation, 
MAX-SNP hardness,
packing linear program.
\end{keywords}

\end{abstract}

\section{Introduction}
\label{sec:introduction}
Consider a classical network application, like search, which requires the assimilation of \textit{source} data available at various servers in order to generate the desired output at a particular server, called the \textit{sink}. This requires the data to be transmitted over the network of communication links connecting the servers and computation of a function of this data. \textit{In-network computation} enables the computation of partial functions of the data on the intermediate servers which may reduce the time (or cost, the number of transmissions) to get the final function value at the \textit{sink}. This situation arises in various other network applications like query processing on a network, and information processing in sensor network, and has been studied extensively, e.g., \cite{Giridhar05,Ying08,Liu13}. In this paper we consider the problem of finding the communication and in-network computation \textit{schedule} of a given arbitrary function of distributed data so as to maximize the \textit{rate} of computation. We give an example to explain our problem below.

\begin{example}
 \label{ex:intro}
Consider a network $\net$ shown in Fig.~\ref{fig:ex_rate}a with capacity of each edge being $1$ bit/second. Each source vertex $s_i$ has an infinite sequence of one bit data $\{x_i(k)\}_{k\geq0}.$ A sink vertex $t$ wants to compute a function $f_t(k)$ of this data where the sequence of computation $(\compgraph)$ is shown by Fig.~\ref{fig:ex_rate}b. Figs.~\ref{fig:ex_rate}c and d show two ways of computing $f_t$ on $net.$ In Fig.~\ref{fig:ex_rate}c all intermediate functions are computed inside $\net$ and $f_t$ is received at $1$ bit/second by $t.$ In Fig.~\ref{fig:ex_rate}d only $\omega_5$ is computed inside $\net$ and  $f_t$ is computed at $0.5$ bits/second rate.\footnote{As the communication link $(a,t)$ is used to transmit both $x_1(k),x_4(k),$ each of them are received at rate $0.5$ bits/second at $t.$} Using both the implementations \footnote{called as \textit{embeddings} in this paper} together, $f_t$ can be computed at $1.5$ bits/second. 
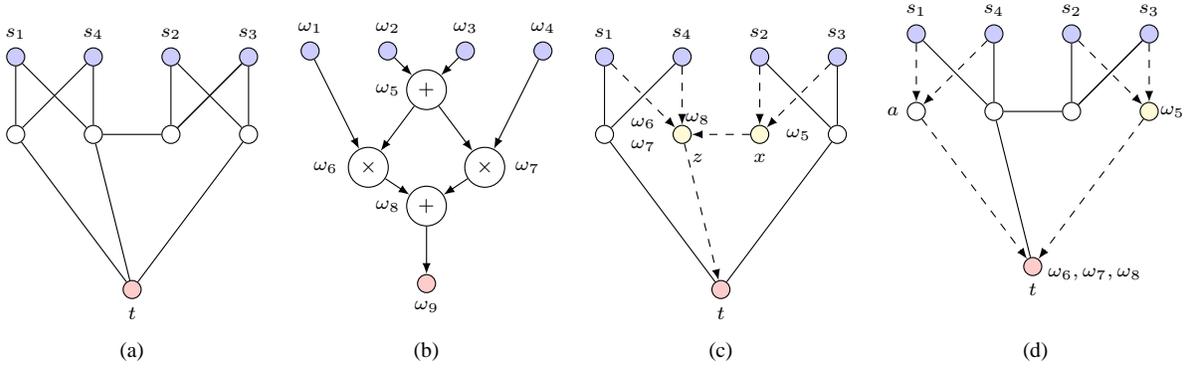
\begin{figure}
\begin{center}
    \resizebox{0.9\linewidth}{!}{
     \subfloat[]{
\begin{tikzpicture}[>=latex]
 \scriptsize
  \tikzstyle{every node} = [circle,draw=black]
 \node (a) at (-1.5,2) [fill=blue!20] {};
 \node (b) at (-0.5,2) [fill=blue!20] {};
 \node (c) at (0.5,2) [fill=blue!20] {};
 \node (d) at (1.5,2) [fill=blue!20] {}; 
 \node (e) at (-1.5,1)  {};
 \node (f) at (-0.5,1)  {};
 \node (g) at (0.5,1)  {};
 \node (h) at (1.5,1) {}; 
 \node (i) at (0,-1) [fill=red!20] {};  
  \node at (-1.5,2.3) [draw=none] {$s_1$};
  \node at (-0.5,2.3) [draw=none] {$s_4$};
   \node at (0.5,2.3) [draw=none] {$s_2$};
    \node at (1.5,2.3) [draw=none] {$s_3$};
    \node at (0,-1.3) [draw=none] {$t$}; 
  \draw[-] (a)-- (e) --(i);
  \draw[-] (a) -- (f) -- (i);
  \draw [-] (b) -- (e);
  \draw [-] (b) -- (f) --(g);
  \draw [-] (c)-- (g) --(d) --(h) --(i);
  \draw [-] (d) --(g);
  \draw[-] (c) -- (h);

 \end{tikzpicture}
}
\subfloat[]{
\begin{tikzpicture}[>=latex]
  \scriptsize
  \tikzstyle{every node} = [circle,draw=black]
 \node (a) at (-1.5,2) [fill=blue!20] {};
 \node (b) at (-0.5,2) [fill=blue!20] {};
 \node (c) at (0.5,2) [fill=blue!20] {};
 \node (d) at (1.5,2) [fill=blue!20] {}; 
 \node (e) at (0,1.5) {$+$};
 \node (f) at (-0.75,0.5) {$\times$};
 \node (g) at (0.75,0.5) {$\times$};
 \node (h) at (0,0) {$+$};
 \node (i) at (0,-1) [fill=red!20] {};  
  \node at (-1.5,2.3) [draw=none] {$\omega_1$};
  \node at (-0.5,2.3) [draw=none] {$\omega_2$};
   \node at (0.5,2.3) [draw=none] {$\omega_3$};
    \node at (1.5,2.3) [draw=none] {$\omega_4$};
    \node at (0,-1.3) [draw=none] {$\omega_9$}; 
    \node at (-0.5,1.5) [draw=none] {$\omega_5$};
     \node at (-1.3,0.5) [draw=none] {$\omega_6$};
      \node at (1.3,0.5) [draw=none] {$\omega_7$};
       \node at (-0.5,0) [draw=none] {$\omega_8$};
 \draw [->] (a) -- (f) ;
 \draw [->] (b) -- (e);
 \draw [->] (c) -- (e);  
 \draw [->] (d) -- (g) ;
 \draw [->] (e) -- (f) ;
 \draw [->] (e) -- (g) ;
 \draw [->] (f) -- (h);
 \draw [->] (g) -- (h);  
 \draw [->] (h) -- (i);

\end{tikzpicture}
} 
\subfloat[]{
\begin{tikzpicture}[>=latex]
 \scriptsize
  \tikzstyle{every node} = [circle,draw=black]
 \node (a) at (-1.5,2) [fill=blue!20] {};
 \node (b) at (-0.5,2) [fill=blue!20] {};
 \node (c) at (0.5,2) [fill=blue!20] {};
 \node (d) at (1.5,2) [fill=blue!20] {}; 
 \node (e) at (-1.5,1)  {};
 \node (f) at (-0.5,1) [fill=yellow!20]  {};
 \node (g) at (0.5,1) [fill=yellow!20] {};
 \node (h) at (1.5,1) {}; 
 \node (i) at (0,-1) [fill=red!20] {};  
  \node at (-1.5,2.3) [draw=none] {$s_1$};
  \node at (-0.5,2.3) [draw=none] {$s_4$};
   \node at (0.5,2.3) [draw=none] {$s_2$};
    \node at (1.5,2.3) [draw=none] {$s_3$};
    \node at (0,-1.3) [draw=none] {$t$}; 
    \node at (-1,1) [draw=none,align=left] {$\omega_6$\\$\omega_7$};
    \node at (-0.3,1.2) [draw=none] {$\omega_8$};
    \node at (1,1) [draw=none] {$\omega_5$};
     \node at (0.5,0.7) [draw=none]{$x$};
    \node at (-0.3,0.7) [draw=none] {$z$};
  \draw[-] (a)-- (e) --(i);
  \draw [-] (b) -- (e);
   \draw [-] (d) -- (h) --(i);
   \draw[-] (c) -- (h);
 
 \draw [->,dashed] (a) -- (f);
 \draw [->,dashed] (f) -- (i);
  \draw [->,dashed] (b) -- (f);
  \draw [->,dashed] (g) -- (f);
  \draw [->,dashed] (c) -- (g);
  \draw [->,dashed] (d) -- (g);
 \end{tikzpicture}
}
 \subfloat[]{
\begin{tikzpicture}[>=latex]
 \scriptsize
  \tikzstyle{every node} = [circle,draw=black]
 \node (a) at (-1.5,2) [fill=blue!20] {};
 \node (b) at (-0.5,2) [fill=blue!20] {};
 \node (c) at (0.5,2) [fill=blue!20] {};
 \node (d) at (1.5,2) [fill=blue!20] {}; 
 \node (e) at (-1.5,1)   {};
 \node (f) at (-0.5,1)  {};
 \node (g) at (0.5,1)  {};
 \node (h) at (1.5,1) [fill=yellow!20] {}; 
 \node (i) at (0,-1) [fill=red!20] {};  
  \node at (-1.5,2.3) [draw=none] {$s_1$};
  \node at (-0.5,2.3) [draw=none] {$s_4$};
   \node at (0.5,2.3) [draw=none] {$s_2$};
    \node at (1.5,2.3) [draw=none] {$s_3$};
    \node at (0,-1.3) [draw=none] {$t$}; 
    \node at (1.8,1) [draw=none] {$\omega_5$};
    \node at (0.8,-1.1) [draw=none,align=left] {$\omega_6,\omega_7,\omega_8$};
    \node at (-1.8,1) [draw=none] {$a$};
  
    \draw[-] (a) -- (f) -- (i);

  \draw [-] (b) -- (f) --(g);
  \draw [-] (c)-- (g) --(d);
  \draw [-] (d) --(g);
  \draw[->,dashed] (c) -- (h);

 \draw[->,dashed] (a) -- (e);
 \draw[->,dashed] (e) -- (i);
 \draw [->,dashed] (b) -- (e);
 \draw[->,dashed] (d) -- (h);
 \draw [->,dashed] (h) --(i);
 \end{tikzpicture}
}
    }
    \caption{(a) Network graph $(\net)$ (b) Computation schema $(\compgraph)$ for $f_t = x_1(x_2+x_3) + x_4(x_2+x_3)$ (c) Implementation 1 computing $f_t$ at $1$ bits/second rate (d) Implementation 2 computing $f_t$ at $0.5$ bits/second rate}
    \label{fig:ex_rate}
\end{center}    
  \end{figure} 
\end{example}

A natural question to ask in this case is that given $\net,\compgraph$ which of all the possible embeddings to compute $f_t$ should one use to get the function at the maximum possible rate and how to schedule the data transfer over the communication links?

\subsection{Maximum Rate Computation Schedule}
\label{sec:rate_relatedwork}

 Recent interest in finding the maximum rate computation schedule is in the context of sensor networks and
distributed computation schemes like MapReduce and Dryad. Computation
of symmetric functions over multihop wireless sensor networks was
introduced in \cite{Giridhar05} and studied in several follow-up works, e.g.,
\cite{Dutta08,Khude05}. More recently, \cite{Kannan13} considered the
computation of such symmetric functions over arbitrary wireline
networks. The objective in the preceding works is, like in this paper, 
maximizing the computation rate. However, they restrict their
attention to symmetric functions which allows them to perform the
computation in an arbitrary order. Further, in
\cite{Giridhar05,Dutta08} the communication network is a random
multihop wireless network and the results are for the asymptotic
regime in the number of sources. While \cite{Kannan13} considers
wireline networks, they obtain outer bound on rate of computation. Authors in \cite{Kannan13} also describe Steiner tree
packing schemes that achieve rates which are are close to this outer bound by showing the approximation factor to be logarithmic
in the number of source nodes. Another line of work, e.g., \cite{Appusamy11,Rai12}, uses network coding techniques to maximize the rate of computation. We do not use network coding in our solution techniques.

The closest to the work in this paper is that of \cite{Shah13,Liu13}
both of which are interested in maximizing the computation rate of
general functions over capacitated networks. In \cite{Shah13},
the computation schema ($\compgraph$) for computing the function $f$ is assumed to be a tree. 
Tree structured $\compgraph$ allows the authors in \cite{Shah13} to obtain the optimum schedule via linear programs that preserve ``functional flow
conservation.''
The functional flow conservation concept of \cite{Shah13} is also used in \cite{Liu13}
when $\compgraph$ is a DAG to find the maximum rate of computation. They give a linear program to find maximum rate of computation and present a distributed algorithm to solve it using Lagrangian dual formulation but do not find the corresponding schedule. The functional flow conservation forces two restrictions on the computation schedule. Firstly, any function can be computed only once in $\net,$ and secondly, every edge of $\compgraph$ should be treated as unique function flow. \footnote{The outgoing edges of vertex $\omega_5$ in Fig.~\ref{fig:ex_rate}b are treated as different flows though they both represent the same function.} These restrictions limit the class of allowable schedules which makes the rate achieved in \cite{Liu13} sub-optimal. 

The problem of \textit{collecting} data at the sink from various sources can be represented by a tree structured computation schema $\compgraph$ where all the source nodes are at the leaves and are connected to the root (acting as sink) directly. Thus an optimal schedule to collect the data at sink can be obtained by using the techniques of \cite{Shah13} which runs in polynomial time in the size of input graphs. This implies that the problem of \textit{optimal data collection} at a single sink is easy to solve. On the other hand, the problem of \textit{distribution} of data from one source to multiple sinks has been studied earlier, e.g., \cite{Jain03} under the name of fractional Steiner tree packing problem. This problem is proved to be MAX SNP-hard \cite{Jain03}. 

In this paper we consider the problem of finding optimal schedule when $\compgraph$ is a general DAG and there is only one sink node in the network.
We first formalize the notion of a schedule to compute a function $f$ over network $\net$ when $\compgraph$ is a DAG which does not have above mentioned restrictions. We define a \emph{routing-computing} scheme (and the rate achieved by it) that computes $f$ in a network (Section~\ref{sec:routingscheme}). 
We show that finding an \textit{optimal routing-computing} scheme is equivalent to finding the solution of a packing linear program of embeddings, which we call capacity achieving linear program (CALP) (Theorem~\ref{thm:equivalnce} in Section~\ref{sec:CALP}).

\subsection{Relating Max Rate to Min Cost Problem}
\label{rate_cost_relation}
Several measures of efficiency of in-network computation like the cost or delay in computation have been studied in the literature \cite{Ying08,Vyavahare14}. These measures may be used when there is only one data value available with each source and the function is computed only once. This is also known as \textit{one shot computation} of the function. In this case the edges of the network graph $\net$ do not represent capacities but have weights associated with them. The weight of an edge corresponds either to the delay incurred or the cost of transmission of a bit between two end points of the edge. 
 The authors in \cite{Vyavahare14} prove that finding \emph{minimum delay embedding} is NP-hard when $\compgraph$ is a DAG and present a polynomial time algorithm when $\compgraph$ is a tree. The problem of finding an \textit{embedding} for one-shot in-network computation which minimizes the cost has been studied under various names in the literature, e.g.,  \cite{Ying08,Bokhari81,Vyavahare14}. 

In this work we relate the complexity of finding the maximum rate schedule to that of finding the \emph{minimum cost embedding}. Specifically, we prove that approximating CALP below a constant factor is NP-hard unless P=NP and even when the degree of each vertex and weights on edges of $\compgraph$ are bounded and $\net$ has just three vertices (Theorem~\ref{thm:main}). This is proved by considering the dual of this LP (Section~\ref{sec:hardness}). We prove that
solving CALP is as hard as solving the separation oracle of its dual (Theorem~\ref{thm:approx_separation}). The separation oracle is a decision problem which reduces to a version of the \emph{minimum cost embedding} problem studied earlier for a different cost model in \cite{Vyavahare14} (defined in Section~\ref{sec:oldcost}). Our cost model comes naturally from the definition of routing-computing scheme for finding the maximum rate (Example~\ref{ex:newcost}).
We prove that our version of \emph{minimum cost embedding} problem is MAX SNP-hard even when $\compgraph$ has bounded degree, bounded edge weights, all outgoing edges of a vertex have the same weight and $\net$ has just three vertices (Corollary~\ref{cor:ratecost_snphardness}). 
We compare our cost model with the one studied in literature \cite{Vyavahare14} and prove that any algorithm which solves the \emph{minimum cost embedding} problem of \cite{Vyavahare14} gives a $D$-approximation for our version of \emph{minimum cost embedding} problem (Theorem~\ref{thm:old_new_relation}) where $D$ is the maximum out-degree of a vertex in $\compgraph.$
 
\subsection{Approximation Algorithms}
\label{sec:results}

As mentioned above, in Theorem~\ref{thm:main} we prove that solving CALP is MAX SNP-hard even when there are only three vertices in $\net.$ Hardness for solving CALP for any network $\net$ with less than three vertices is of theoretical interest. Thus, we first present a polynomial time procedure to solve CALP on $\net$ with two vertices for an arbitrary DAG $\compgraph$ (Section~\ref{sec:2nodenet_algo}) thus proving the dichotomy of hardness of CALP.

In Section~\ref{sec:app_algo} we present a restricted class of schedules by studying a restricted class of embeddings, called \rembedding. We present the equivalent packing LP for these embeddings called R-CALP and observed that our hardness results (Theorem~\ref{thm:approx_separation} and Theorem~\ref{thm:main}) also hold for this class of schedules. We use the procedure of Theorem~\ref{thm:approx_separation} in Section~\ref{sec:app_algo} to present approximation algorithms for R-CALP. Using the relation derived in Theorem~\ref{thm:old_new_relation} between different cost models and the result of \cite{Karloff06} we show that there is no polynomial time constant factor approximation for R-CALP (Corollary~\ref{cor:oldcost_hardness}) unless $NP \subseteq DTIME(p^{poly(\log p)})$ when $\compgraph$ has unbounded degree and edge weights. Here $p$ is the number of vertices in $\compgraph.$ 

Since the problem for general $\compgraph$ is NP-hard, we consider some specific structures of $\compgraph$ to get approximate algorithms. Many of the well known functions like fast Fourier transform (FFT), sorting or any polynomial function of input data can be represented by a layered computation graph. We present a constant factor approximate algorithm for R-CALP when the width of each layer of the layered computation graph is bounded (Corollary~\ref{cor:boundedwidth_layer}). Then we consider a class of $\compgraph$ that has a spanning tree such that any edge is a part of at most $O(F)$ fundamental cycles. For a $N$ point FFT computation graph $F= \log(N).$ We present a polynomial time $O(FD)$-approximation algorithm to solve R-CALP for such graphs (Corollary~\ref{cor:unboundedwidth_layer}). 
Lastly we formulate the \emph{minimum cost embedding} problem as a quadratic integer program and present its linear programming relaxation based on \emph{earthmover distance metric} (Section~\ref{sec:lp}). Applying the randomized rounding techniques to the optimal solution of this LP we present two algorithms (derived from \cite{Calinescu01}) to approximate R-CALP. The first algorithm gives an $O(D\log n)$-approximation for general $\compgraph$ (Corollary~\ref{cor:logn-approx}) and the second algorithm gives an $O(D)$-approximation for planar $\compgraph$ (Corollary~\ref{cor:planar-approx}) where $n$ is the number vertices in $\net.$

\section{Notations and Problem Definition}
\label{sec:problemdef}

A communication network is represented by an undirected graph
$\net=(\netnodes,\netedges)$, where $V=\{u_1,\ldots,u_n\}$ is a set of network nodes and $E$ is a set of communication links (see Fig.~\ref{fig:example}a for an example of $\net.$) Each link has a non-negative capacity associated with it. Let
$\{s_1,s_2, \ldots,s_\kappa\} \subset \netnodes$ be the set of
$\kappa$ source nodes with $s_i$ generating an infinite sequence of
data values from the alphabet $\mathcal{A}_i.$ The sink node $t$ needs to compute
function $f:\{\mathcal{A}_1 \times \mathcal{A}_2 \times \cdots, \times
\mathcal{A}_{\kappa}\} \mapsto \mathcal{A}_t.$ The schema
to compute $f$ is given as a directed acyclic graph $\compgraph =
(\compnodes,\compedges)$ where $\compnodes$ is the set of nodes
representing a computation of an intermediate (with respect to $f$) function of the
data and $\compedges$ is the set of edges denoting the
communication of these functions. Let $\{\omega_1,\omega_2,
\ldots,\omega_{\kappa}\} \subset \Omega$ be the source nodes and $\omega_p$ be
the sink that receives $f(\cdot).$ See Fig.~
\ref{ex:embedding}b for an example of $\compgraph.$ 

Let $\{x_i(k)\}_{k \geq 1}$ be the infinite sequence of data values at source $s_i$. 
We assume that the entire sequence is available at $s_i$ all the
time. Let $f_t(k):= f(x_1(k), \ldots, x_{\kappa}(k)).$ Our interest in
this paper is in the computation and communication schedule in $\net$
that will obtain $f_t(k)$ at sink node $t$ at the maximum rate.
The source nodes of $\compgraph$ have in-degree zero while
out-degree of sink node $\omega_p$ is zero. All the other nodes
$\compgraph$ have in-degree greater than zero and out-degree greater
than zero\footnote{If the out-degree of all the nodes (except the sink
  node which has out-degree zero) is strictly one then the graph
  $\compgraph$ is a tree structure.}. The direction on the edges in
$\compgraph$ represents the direction of the data flow. Without loss of generality we assume that all the outgoing edges of a
node represent the same intermediate function.  Let
$\compedges_\theta$ be the set of all edges carrying the 
intermediate function $\theta$ and let $\mathcal{A}_\theta$ be its
(finite) alphabet. Let $\Theta$ be the set of all intermediate
functions in $\compgraph$, let $w : \Theta \mapsto \mathbb{Z}^{+}$ be the
weight of each intermediate function in $\compgraph$ with $w(\theta)
= \lceil \log(|\mathcal{A}_{\theta}|) \rceil.$ 
\begin{remark}
  \label{rmk:out-func}
  Each outgoing edge of any vertex $\omega \in
  \compnodes$ carries the same function, the weights associated with
  all the outgoing edges of a given $\omega$ are the same.
\end{remark}

A path in $\net$ is denoted by a sequence of distinct vertices $\sigma
=(u_1,u_2,\ldots,u_l)$, such that $(u_i,u_{i+1}) \in \netedges$ $
~\forall 1 \leq i \leq l-1 $. The nodes $u_1$ and $u_l$ are called the
start node ($\pathstart \sigma$) and the end node ($\pathend \sigma$)
of the path $\sigma$ respectively. A path can be of zero length in
which case $\sigma = (u_1)$ is a single vertex and start and end nodes
are the same. $\setofpaths$ is the set of all paths in $\net.$
For $\gamma \in \compedges$ let $\tail\gamma$ and $\head \gamma$
represent the head and the tail of the edge $\gamma$ respectively. Let
$\preedges \gamma$ and $\sucedges \gamma$ denote, respectively, the
immediate predecessors and successors of $\gamma,$ i.e., $ \preedges
\gamma = \{\alpha \in \compedges | \head \alpha = \tail \gamma\}$ and
$ \sucedges \gamma = \{\alpha \in \compedges | \tail \alpha = \head
\gamma\}$.
For a function $\theta \in \Theta,$ let $\prefunction \theta$ and
$\sucfunction \theta$ be the functions carried by the predecessor and
successor edges of $\Gamma_{\theta}.$

\subsection{Embedding Definition}
\label{sec:embedding-def}
Informally an
embedding of $\compgraph$ on $\net$ gives a way of computing $f$ on $\net$ as per the data flow given by $\compgraph$. Thus, an embedding of $\compgraph$ on $\net$ can be seen as a function which maps an edge $\gamma \in \compedges$ to paths in $\net$ where the the function carried by $\gamma$ is computed at the start node of the path and at the end node of the path it is used to generate its successor function. This is formalized in the following definition.

\begin{definition} [\textbf{Embedding}]
  \label{def:embedding}
  An embedding of $\compgraph$ on $\net$ is a function $\mathcal{E}: \compedges \mapsto
  \mathcal{P}(\setofpaths).$\footnote{Here $\mathcal{P}(\Sigma)$
    denotes the power set of $\Sigma$ except the empty set. In an
    embedding an edge may get mapped to a path of zero length, which
    implies that both its end points are mapped to the same vertex.
  } 
  If $\mathcal{E}(\gamma_l) :=
  \{\sigma_1^l,\ldots,\sigma_r^l\}$ then the edge $\gamma_l$ is mapped
  to $r$ paths such that the following properties are satisfied.
  \begin{enumerate}
  \item If $\tail {\gamma_l} = \omega_i, \forall i \in [1,\kappa]$ then $\pathstart{\sigma_a^l} = s_i$ $ \forall \sigma_a^l \in \mathcal{E}(\gamma_l).$ 
  \item If $\head{\gamma_l} = \omega_p$ then $\pathend{\sigma_a^l} = t$ $\forall \sigma_a^l \in
    \mathcal{E}(\gamma_l).$
  \item If $\gamma_i \in \sucedges{\gamma_j}$ then there exists a $\sigma_b^j$ such that
    $\pathend{\sigma_b^j} = \pathstart{\sigma_a^i}$ $\forall \sigma_a^i.$ Similarly, for
    every $\sigma_b^j$ there exists a $\sigma_a^i$ such that
    $\pathend{\sigma_b^j} = \pathstart{\sigma_a^i}.$
  \item There are no $i,j \in
    [1,r]$ such that $i\neq j$ and $\pathend{\sigma_i^l} =
    \pathend{\sigma_j^l}$ $\forall \gamma_l \in \compedges.$ 
    \item If $\pathstart{\sigma_i^l} \neq
    \pathstart{\sigma_j^l}$ $\forall i \neq j \in
    [1,r]$ then $\sigma_i^l \cap \sigma_j^l =
    \emptyset$ $\forall \gamma_l \in \compedges.$ 
 \end{enumerate}
\end{definition}

Above mentioned properties of a valid embedding are a direct consequence of the structure of $\compgraph$ which are explained in Appendix~\ref{app:embedding_property}.

  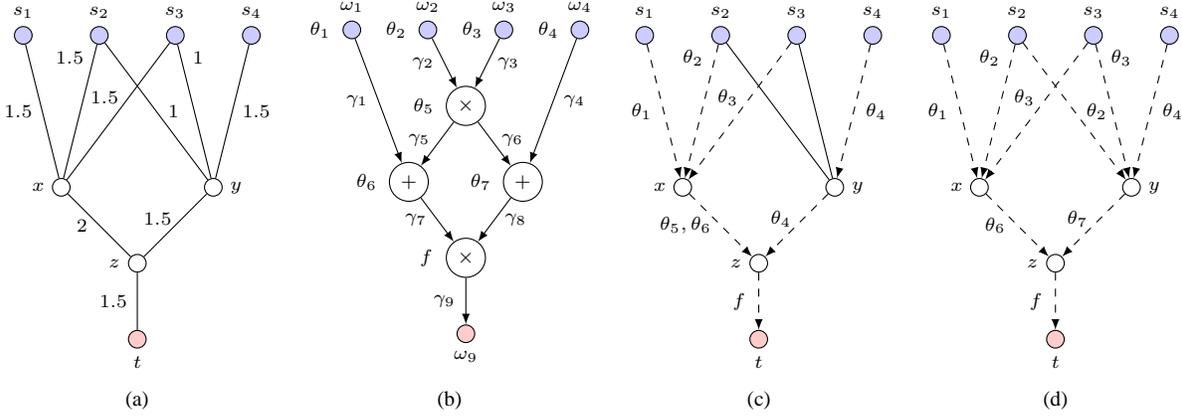
\begin{figure}[t]
  \begin{center}
    \resizebox{0.9\linewidth}{!}{
      \subfloat[]{
\begin{tikzpicture}[>=latex]
\scriptsize
 \tikzstyle{every node} = [circle,draw=black]
 \node (a) at (-1.5,2) [fill=blue!20] {};
 \node (b) at (-0.5,2) [fill=blue!20] {};
 \node (c) at (0.5,2) [fill=blue!20] {};
 \node (d) at (1.5,2) [fill=blue!20] {}; 
 \node (e) at (-1,0) {};
 \node (f) at (1,0) {};
 \node (g) at (0,-1) {};
 \node (h) at (0,-2) [fill=red!20] {};
 \node at (-1.5,2.3) [draw=none] {$s_1$};
  \node at (-0.5,2.3) [draw=none] {$s_2$};
   \node at (0.5,2.3) [draw=none] {$s_3$};
    \node at (1.5,2.3) [draw=none] {$s_4$};
    \node at (0,-2.3) [draw=none] {$t$}; 
     \node at (-1.3,0) [draw=none] {$x$};
      \node at (1.3,0) [draw=none] {$y$};
       \node at (-0.3,-1) [draw=none] {$z$};
    
 \draw [-] (a) -- (e) node [draw=none,midway,left] {$1.5$} ;
 \draw [-] (b) -- (e) node [draw=none,pos=0.1,left] {$1.5$};
 \draw [-] (b) -- (f)node [draw=none,pos=0.5,right] {$1$};
 \draw [-] (c) -- (e)node [draw=none,pos=0.4,left] {$1.5$};  
 \draw [-] (c) -- (f) node [draw=none,pos=0.1,right] {$1$};
 \draw [-] (d) -- (f) node [draw=none,midway,right] {$1.5$};
 \draw [-] (e) -- (g) node [draw=none,midway,left] {$2$};
 \draw [-] (f) -- (g) node [draw=none,pos=0.4,left] {$1.5$};
 \draw [-] (g) -- (h)node [draw=none,midway,left] {$1.5$}; 

 \end{tikzpicture}
}

\subfloat[]{
\begin{tikzpicture}[>=latex]
  \scriptsize
  \tikzstyle{every node} = [circle,draw=black]
 \node (a) at (-1.5,2) [fill=blue!20] {};
 \node (b) at (-0.5,2) [fill=blue!20] {};
 \node (c) at (0.5,2) [fill=blue!20] {};
 \node (d) at (1.5,2) [fill=blue!20] {}; 
 \node (e) at (0,1) {$\times$};
 \node (f) at (-0.75,0) {$+$};
 \node (g) at (0.75,0) {$+$};
 \node (h) at (0,-1) {$\times$};
 \node (i) at (0,-2) [fill=red!20] {};  
  \node at (-1.5,2.3) [draw=none] {$\omega_1$};
  \node at (-0.5,2.3) [draw=none] {$\omega_2$};
   \node at (0.5,2.3) [draw=none] {$\omega_3$};
    \node at (1.5,2.3) [draw=none] {$\omega_4$};
    \node at (0,-2.3) [draw=none] {$\omega_9$}; 

 \draw [->] (a) -- (f) node [midway,left,draw=none] {$\gamma_1$};
 \draw [->] (b) -- (e) node [midway,left,draw=none] {$\gamma_2$};
 \draw [->] (c) -- (e) node [midway,right,draw=none] {$\gamma_3$};  
 \draw [->] (d) -- (g) node [midway,right,draw=none] {$\gamma_4$};
 \draw [->] (e) -- (f) node [pos=0.4,left,draw=none] {$\gamma_5$};
 \draw [->] (e) -- (g) node [pos=0.4,right,draw=none] {$\gamma_6$};
 \draw [->] (f) -- (h)node [midway,left,draw=none] {$\gamma_7$};
 \draw [->] (g) -- (h) node [midway,right,draw=none] {$\gamma_8$};  
 \draw [->] (h) -- (i) node [midway,left,draw=none] {$\gamma_9$};
 \node [anchor=east,draw=none] at (a.west) {$\theta_1$};
 \node [anchor=east,draw=none] at (b.west) {$\theta_2$};
 \node [anchor=east,draw=none] at (c.west) {$\theta_3$};
 \node [anchor=east,draw=none] at (d.west) {$\theta_4$};
 \node [anchor=east,draw=none] at (e.west) {$\theta_5$};
 \node [anchor=east,draw=none] at (f.west) {$\theta_6$};
 \node [anchor=east,draw=none] at (g.west) {$\theta_7$};
 \node [anchor=east,draw=none] at (h.west) {$f$};
\end{tikzpicture}
} 
 \subfloat[]{
\begin{tikzpicture}[>=latex]
\scriptsize
 \tikzstyle{every node} = [circle,draw=black]
 \node (a) at (-1.5,2) [fill=blue!20] {};
 \node (b) at (-0.5,2) [fill=blue!20] {};
 \node (c) at (0.5,2) [fill=blue!20] {};
 \node (d) at (1.5,2) [fill=blue!20] {}; 
 \node (e) at (-1,0) {};
 \node (f) at (1,0) {};
 \node (g) at (0,-1) {};
 \node (h) at (0,-2) [fill=red!20] {};
  \node at (-1.5,2.3) [draw=none] {$s_1$};
  \node at (-0.5,2.3) [draw=none] {$s_2$};
   \node at (0.5,2.3) [draw=none] {$s_3$};
    \node at (1.5,2.3) [draw=none] {$s_4$};
    \node at (0,-2.3) [draw=none] {$t$}; 
     \node at (-1.3,0) [draw=none] {$x$};
      \node at (1.3,0) [draw=none] {$y$};
       \node at (-0.3,-1) [draw=none] {$z$};
 \draw [->,dashed] (a) -- (e) node [draw=none,midway,left] {$\theta_1$} ;
 \draw [->,dashed] (b) -- (e) node [draw=none,pos=0.1,left] {$\theta_2$};
 \draw [-] (b) -- (f);
 \draw [->,dashed] (c) -- (e)node [draw=none,pos=0.4,left] {$\theta_3$};  
 \draw [-] (c) -- (f) ;
 \draw [->,dashed] (d) -- (f) node [draw=none,midway,right] {$\theta_4$};
 \draw [->,dashed] (e) -- (g) node [draw=none,midway,left] {$\theta_5,\theta_6$};
 \draw [->,dashed] (f) -- (g) node [draw=none,pos=0.4,left] {$\theta_4$};
 \draw [->,dashed] (g) -- (h)node [draw=none,midway,left] {$f$};  
  \end{tikzpicture}
}
 \subfloat[]{
\begin{tikzpicture}[>=latex]
\scriptsize
 \tikzstyle{every node} = [circle,draw=black]
 \node (a) at (-1.5,2) [fill=blue!20] {};
 \node (b) at (-0.5,2) [fill=blue!20] {};
 \node (c) at (0.5,2) [fill=blue!20] {};
 \node (d) at (1.5,2) [fill=blue!20] {}; 
 \node (e) at (-1,0) {};
 \node (f) at (1,0) {};
 \node (g) at (0,-1) {};
 \node (h) at (0,-2) [fill=red!20] {};
  \node at (-1.5,2.3) [draw=none] {$s_1$};
  \node at (-0.5,2.3) [draw=none] {$s_2$};
   \node at (0.5,2.3) [draw=none] {$s_3$};
    \node at (1.5,2.3) [draw=none] {$s_4$};
    \node at (0,-2.3) [draw=none] {$t$}; 
     \node at (-1.3,0) [draw=none] {$x$};
      \node at (1.3,0) [draw=none] {$y$};
       \node at (-0.3,-1) [draw=none] {$z$};
 \draw [->,dashed] (a) -- (e) node [draw=none,midway,left] {$\theta_1$} ;
 \draw [->,dashed] (b) -- (e) node [draw=none,pos=0.1,left] {$\theta_2$};
 \draw [->,dashed] (b) -- (f)node [draw=none,pos=0.5,right] {$\theta_2$};
 \draw [->,dashed] (c) -- (e)node [draw=none,pos=0.4,left] {$\theta_3$};  
 \draw [->,dashed] (c) -- (f) node [draw=none,pos=0.1,right] {$\theta_3$};
 \draw [->,dashed] (d) -- (f) node [draw=none,midway,right] {$\theta_4$};
 \draw [->,dashed] (e) -- (g) node [draw=none,midway,left] {$\theta_6$};
 \draw [->,dashed] (f) -- (g) node [draw=none,pos=0.4,left] {$\theta_7$};
 \draw [->,dashed] (g) -- (h)node [draw=none,midway,left] {$f$}; 

 \end{tikzpicture}
}
    }
    \caption{(a). Network graph $(\net)$ Number near an edge shows
      its capacity in bits/second (b). Computation graph
      $(\compgraph)$ for $f = (x_1+x_2x_3)(x_4+x_2x_3)$ (c). An
      embedding $\embedding_1$  of function $f$ on $\net$ (d). Another
      embedding $\embedding_2$ to computer $f$ }
    \label{fig:example}
 \end{center}   
  \end{figure}
\begin{example}
  \label{ex:embedding}
  Consider $\net= (\netnodes,\netedges)$ as shown in
  Fig.~\ref{fig:example}a. 
  Assume that each source generates symbols from 
 $\mathcal{A} = \{0,1\}$ and
  the alphabet of function $f$ is also $\mathcal{A}$. A
  schema $\compgraph$ to compute the function $f$ is shown in
  Fig.~\ref{fig:example}b.
  Assume that all the intermediate functions are also from
  $\mathcal{A}$, hence $w(\theta) = \lceil \log(2)
  \rceil =1$ for all $\theta \in \Theta.$
  Two of the (multiple)
  possible embeddings are shown in the Fig.~\ref{fig:example}c and
  d. For the embedding shown in Fig~\ref{fig:example}c,
  $\embedding_1(\gamma_1) = s_1x, \embedding_1(\gamma_2) = s_2x,
  \embedding_1(\gamma_3) = s_3x, \embedding_1(\gamma_4) = s_4yz,
  \embedding_1(\gamma_5) = x, \embedding_1(\gamma_6) = xz,
  \embedding_1(\gamma_7) = xz, \embedding_1(\gamma_8) = z,
  \embedding_1(\gamma_9) = zt.$ For the embedding shown in
  Fig~\ref{fig:example}d, $\embedding_2(\gamma_1) = s_1x,
  \embedding_2(\gamma_2) = \{s_2x,s_2y\}, \embedding_2(\gamma_3) =
  \{s_3x,s_3y\}, \embedding_2(\gamma_4) = s_4y,$ $ \embedding_2(\gamma_5)
  = x, \embedding_2(\gamma_6) = y, \embedding_2(\gamma_7) = xz,$ $
  \embedding_2(\gamma_8) = yz,$ $ \embedding_2(\gamma_9) = zt.$

\end{example}

Observe that if an edge $\gamma_l$ is mapped to two paths, say
$\sigma_1^l$ and $\sigma_2^l,$ then the same symbol of the function carried
by it is generated twice; once by the vertex $\pathstart {\sigma_1^l}$
and once by vertex $\pathstart {\sigma_2^l}.$ We denote the set of all
the embeddings of $\compgraph$ on $\net$ by $\setofembeddings.$
As observed in Example~\ref{ex:embedding}, an edge in $\net$ can either
carry zero or more function types in an embedding. Let
$r_\mathcal{E}^\theta(e) := \mathbbm{1}\{e \in \sigma_i^l|\sigma_i^l
\in \mathcal{E}(\gamma_l) \mbox{ and } \gamma_l \in \Gamma_\theta\}$
be the indicator function of the transmission of function type
$\theta$ over an edge $e \in \netedges.$ Then total number of times an edge is used
in $\embedding$ is  $r_\mathcal{E}(e) :=
\sum\limits_{\theta \in \Theta} r_\mathcal{E}^\theta(e) w(\theta).$ 

\begin{remark}
An edge $e$ in $\net$ can be a part of embedding of more than one edges of $\compgraph$ all of which carry the same function $\theta.$ In this case we say that the edge $e$ is used only once (observe $r_\mathcal{E}^\theta(e))$ since the edges carry the same function.
\end{remark}

The notion of an embedding of $\compgraph$ on $\net$ to compute $f$ is used
in \cite{Shah13,Liu13}. The key difference between these and this
paper is that in the former, an edge in $\compgraph$ is mapped to only
one path in $\net.$ This is not a restriction when
$\compgraph$ is a tree, like in \cite{Shah13}. However, it does reduce
the maximum rate when $\compgraph$ is a DAG as demonstrated by  the following example.

\begin{example}
\label{ex:embedding-rate}
We continue with Example~\ref{ex:embedding} here. Observe that in $\embedding_2$ (shown in Fig.~\ref{fig:example}d) the function $\theta_5$ is computed at two vertices $x$ and $y$ and used to compute $\theta_6$ at $x$ and $\theta_7$ at $y.$ 
The source $s_2$ sends the function $\theta_2$ on $s_2x,s_2y$ and $s_3$ sends $\theta_3$ on $s_3x,s_3y.$ If the capacity of links $s_2y$ and $s_3y$ are used completely the final function $f$ can be computed at the rate of $1$ bits per second using $\embedding_2.$ As each edge in $\net$ is used only once, $r_{\embedding_2}(e) = 1 ~\forall e \in \netedges.$ 

Note that after the usage of edges by $\embedding_2$ residual capacities on the edges of $\net$ are: $c(s_1x) = 0.5,c(s_2x) =0.5,c(s_2y)=0,c(s_3x)=0.5,c(s_3y)=0,c(s_4y)=0.5,c(xz)=1,c(yz)=0.5$ and $c(zt)=0.5.$ These residual capacities can be used by $\embedding_1$ (shown in Fig~\ref{fig:example}c) to generate the function $f$ at rate $0.5$ bits/second. Note that for all the edges used by $\embedding_1,$ $r_{\embedding_1}(e) = 1$ except for $xz$ for which $r_{\embedding_1}(xz) =2.$ Using both the embeddings, the sink $t$ can receive $f$ at the rate of $1.5$ bits/second. 
\end{example}

\subsection{Communication and Computation Model}
\label{sec:routingscheme}
We saw that an embedding of $\compgraph$ on $\net$ specifies which
function $\theta$ is generated at which vertex and transmitted over
which edge in the network. However, this does not specify the exact
schedule for computing each $\theta.$ Our task is to not only give an
embedding but also give a full schedule. For this we define the notion
of \emph{routing-computing scheme.}

To define the scheme formally, we first mention the assumptions on the computation of functions and the allowed set of communication events in the network graph.
 Let $\mathbb{X}$ denote the
vector $[x_1,\ldots,x_{\kappa}],$ and its $k-$th realization be
$\mathbb{X}(k) = [x_1(k),\ldots,x_{\kappa}(k)].$ 
The time is slotted and in each time
slot an edge $e= (u,v) \in \netedges$ is said to be activated if some
information is transferred from $u$ to $v.$ All the edges can
be activated simultaneously in any time slot. If the capacity of an
edge $e$ is $c(e)$ then at most $\lfloor c(e) T \rfloor$ 
bits can be transferred over it in $T$ time slots. 
We assume that any vertex $u$ transmits all the
bits of the $k$-th realization of function $\theta$ on the edge $e$ as
a single packet of $w(\theta)$ bits. Any $u \in \netnodes$ at time slot $\tau$ may
perform one of the following tasks exclusively.
\begin{enumerate}
\item \emph{Computation event}: if there exists 
$\tau' <\tau$ such that the $k$-th realization of the predecessor functions
of $\theta$ are received or generated by $u$ then it can generate the $k$-th
realization of $\theta.$ 
\item \emph{Communication event}: if there exists  $\tau' < \tau$
such that the $k$-th realization of a function $\theta$ was either
received or generated by $u$ then it can transmit it over one of its
outgoing edges, say $(u,v).$ 
\item Receive a function from an
incoming edge or do nothing.
\end{enumerate}
We assume that any computation event in the network can happen instantaneously and the time is taken into consideration only for communication events (which is dictated by the capacity of network edges as mentioned above). Any routing-computing scheme can be considered as a sequence of $L$
events $R_l, 1 \leq l \leq L$ where each event is one of above
mentioned tasks. It computes $K$ symbols of $f$ at the sink in time $t$ by using $K$ fixed block of source symbols indexed by $1,2,\ldots,K.$ The rate of computation of $f$ by the routing-computing scheme is then defined as $K/t.$ 
At any time $\tau \leq t,$ a node can have, a subset of the universe of data
$\mathcal{U} = \Theta \times [1,K],$ where an element $(\theta,k) \in
\mathcal{U}$ denotes the $k$-th symbol of the function $\theta.$
The sets $\mathcal{U}_{u,l},\mathcal{U}_{u,l+1}
\subseteq \mathcal{U}$ represent the state of a node $u$ before and
after the $l$-th event $R_l$ respectively.
In the case of a computation event the state of only $u$ is changed, and
for a
communication event only the states of vertices $u$ and $v$ are
changed. As seen in Example~\ref{ex:embedding}, 
a symbol of a function can be computed multiple times in the
network and the scheme presented here takes this into account.
Let  $m_{u,k}^{\theta}$ be
  the number of times the $k$-th symbol of $\theta$ is used
  or transmitted by $u$ in the overall scheme. We remind you that when $\compgraph$ is a tree, each function symbol is computed
only once in the network and the corresponding scheme is presented in \cite{Shah13}.

\begin{definition}
\label{def:routing-scheme}
  A $(\{N_e|e \in \netedges\},K,m_{u,k}^{\theta})$ routing-computing
  scheme for $(\net,\compgraph)$ given  $L \in \mathbb{N}^+,$
  subsets $\{\mathcal{U}_{u,l} \subseteq \mathcal{U}| u \in
  \netnodes, l \in [1,L+1]\}$ and 
  $\forall u,k,\theta: m_{u,k}^{\theta} \in \mathbb{N}^+$ is:
  \begin{enumerate}
  \item For $1\leq i\leq \kappa,$ $\mathcal{U}_{s_i,1} =
    \{(\theta_i,k) | k \in [1,K]\}$, $\mathcal{U}_{u,1} =
    \emptyset ~\forall u \in \netnodes \setminus \{s_i|1 \leq i \leq
    \kappa\}.$
  \item For each $l <L+1,$ one of the following holds.

    \begin{enumerate}
    \item \textit{Computation event:} In this event a node $u$
      computes a function $\theta(\mathbb{X}(k))$ using
      $\{\eta(\mathbb{X}(k))|\eta \in \prefunction \theta\}.$ More
      precisely we first set $m_{u,k}^{\eta}=m_{u,k}^{\eta} -1
      ~\forall \eta \in \prefunction \theta$ and $Z(\mathcal{U}_{u,l})
      := \{(\gamma,k) \in \mathcal{U}_{u,l} | m_{u,k}^{\gamma} =
      0\}$. Then the data-sets are updated as follows:
      $\mathcal{U}_{u,l+1} = \{(\theta,k)\} \cup \mathcal{U}_{u,l}
      \setminus Z(\mathcal{U}_{u,l}); \mathcal{U}_{v,l+1} =
      \mathcal{U}_{v,l}, ~\forall v \in \netnodes \setminus \{u\}. $
     
    \item \textit{Communication event:} In this event a function
      $\theta(\mathbb{X}(k))$ is transmitted on the link $uv.$ More
      precisely we first set $m_{u,k}^{\theta} = m_{u,k}^{\theta} -1$
      and $Z(\mathcal{U}_{u,l}) := \{(\gamma,k) \in \mathcal{U}_{u,l}
      | m_{u,k}^{\gamma} = 0\}.$ Then the data-sets are updated as
      follows: $\mathcal{U}_{v,l+1} = \mathcal{U}_{v,l} \cup
      \{(\theta,k)\} ; \mathcal{U}_{u,l+1} = \mathcal{U}_{u,l}
      \setminus Z(\mathcal{U}_{u,l}); \mathcal{U}_{w,l+1} =
      \mathcal{U}_{w,l} ~\forall w \neq u,v.$

    \item \textit{Final condition:} $\mathcal{U}_{t,L+1} = \{(f,k)| 1
      \leq k \leq K\}; \mathcal{U}_{u,L+1} = \emptyset ~\forall u \neq
      t ; m_{u,k}^{\theta} = 0 ~\forall u \in \netnodes, k \in [1,K],
      \theta \in \Theta.$

   \item \textit{Total link usage:} Let $r_e^{\theta}$ 
     be the number of times a function $\theta$ is transmitted over
     edge $e\in \net.$ Then the total link usage is given by: $ N_e =
     \sum_{\theta \in \Theta} r_e^{\theta} w(\theta).$
  \end{enumerate}
\end{enumerate} 
 
\end{definition}
The scheme uses an edge $e \in \netedges$ for
$N_e/c(e)$ time slots to compute $K$ symbols of $f$ at the sink.
\begin{definition}
  For a given network $\net,$ $\{c(e)|e \in
  \netedges\},$ and a computation graph $\compgraph$, a
  rate $\lambda$ is said to be $(\net,\compgraph)$-achievable if for
  every $\epsilon >0,$ there is a $(\{N_e|e \in \netedges\},K,m_{u,k}^{\theta})$
  routing-computing scheme for $(\net,\compgraph)$ such that
  $N_e(\lambda-\epsilon) \leq Kc(e), ~\forall e \in \netedges.$ The
  supremum of $(\net,\compgraph)$-achievable rates over all the
  routing-computing schemes is called the computing capacity for
  $(\net,\compgraph),$ and is denoted by $C(\net,\compgraph).$ \footnote{A similar definition appears in~\cite{Shah13}, however in their case $\compgraph$ is a tree.}
\end{definition}

Example~\ref{ex:embedding-rate} presented in
Section~\ref{sec:embedding-def} shows that using multiple embeddings and
sequencing them appropriately we can achieve a higher
rate of function computation than by just using one embedding. In the
next section we give a (packing) linear program for obtaining maximum
rate of computation using a combination of different embeddings and
show that this also achieves the computing capacity
$C(\net,\compgraph).$

\section{Capacity Achieving LP (CALP)}
\label{sec:CALP}

\centerline{\rule{\columnwidth}{0.75pt}}
\textbf{Capacity Achieving Linear Program (CALP)}

\textbf{Objective:} Maximize $R:= \sum_{\embedding \in \setofembeddings} x(\embedding)$
\textbf{subject to}
\begin{enumerate}
\item Capacity constraints: $\sum_{\embedding \in
    \setofembeddings} r_{\embedding}(e)x(\embedding) \leq c(e), \mbox{
  } ~\forall e \in \netedges. $

 \item Non-negativity constraints: $x(\embedding) \geq 0, ~\forall \embedding \in \setofembeddings.$

\end{enumerate}

\centerline{\rule{\columnwidth}{0.75pt}}

\begin{theorem}
 \label{thm:equivalnce}
For a given network $\net$ and computation DAG $\compgraph,$ CALP achieves a rate $R$ which is equal to the computing capacity ($C(\net,\compgraph))$ for $(\net,\compgraph).$
\end{theorem}
\begin{proof}
We prove the theorem in two steps. First we show achievability, i.e., we show that for any $\{x(\mathcal{E})|\mathcal{E} \in \mathbb{E}\}$ that satisfies the constraints of the CALP the rate $\sum\limits_{\mathcal{E} \in \mathbb{E}} x(\mathcal{E})$ is $(\net,\compgraph)-$achievable. Next we show that for any $(\{N_e|e \in \netedges\},K,m_{u,k}^{\theta})$ routing-computing scheme for $(\net,\compgraph)$ satisfying $N_e\lambda \leq Kc(e), ~\forall e \in \netedges$ there exists $\{x(\mathcal{E})|\mathcal{E} \in \mathbb{E}\}$ satisfying the constraints of the CALP such that $\sum\limits_{\mathcal{E} \in \mathbb{E}} x(\mathcal{E}) = \lambda.$ Authors in \cite{Shah13} defined routing-computing scheme works only for tree structured $\compgraph$  where any intermediate function is computed only once in the network and showed its equivalence to the corresponding CALP using similar arguments.

\textbf{Step $1$ of the proof:} In this step starting with a set of embeddings which satisfies the CALP constraints we generate a routing-computing scheme which achieves the sum rate of these embeddings. Let $\{x(\embedding)|\embedding \in \setofembeddings\}$ be the number of symbols of function $f$ generated by various embeddings such that it satisfies the constraints of CALP.  Since the rational numbers are dense we can find a set of rational flows $\{x'(\mathcal{E})| \mathcal{E} \in \mathbb{E}\}$ such that $\sum_{\mathcal{E} \in \mathbb{E}} x'(\mathcal{E}) \geq \sum_{\mathcal{E} \in \mathbb{E}} x(\mathcal{E}) - \epsilon$ for any $\epsilon >0.$ We denote the least common multiple of the denominators of $\{x'(\mathcal{E})| \mathcal{E} \in \mathbb{E}\}$ by $d.$ 
 Let us take $K = d \sum_{\mathcal{E} \in \mathbb{E}} x'(\mathcal{E}).$ For every edge $e \in \netedges$ let $N_e = d \sum_{\mathcal{E} \in \mathbb{E}} r_{\mathcal{E}}(e) x'(\mathcal{E}).$ 
  An embedding tells us where any function is computed in the network and on which edges it is transmitted. Let $L(\mathcal{E}) = \sum\limits_{e \in \netedges}  \sum\limits_{\theta \in \Theta} r_\mathcal{E}^{\theta} (e)$ denote the number of symbols of different functions transmitted in the embedding $\mathcal{E},$ where $r_\mathcal{E}^{\theta} (e)$ is the indicator variable for the transmission of function type $\theta$ over edge $e$ in embedding $\embedding.$ Similarly let $g_{\embedding}(\theta)$ be the number of times a function $\theta \in \{\Theta \setminus \{x_u | i \in [1,\kappa]\}\}$ is computed under the embedding $\mathcal{E}.$  More formally,
 \begin{displaymath}
 g_{\mathcal{E}}(\theta) :=  \sum\limits_{\gamma_1,\gamma_2 \in \Gamma_\theta}
  \mathbbm{1} \{\pathstart{\sigma_i} \neq \pathstart{\sigma_j} | \forall \sigma_i \in \mathcal{E}(\gamma_1) \mbox{ and } \sigma_j \in \mathcal{E}(\gamma_2) \}. \footnote{Note that in the above equation we need to consider all the values of $\gamma_1$ and $\gamma_2$ including $\gamma_1 = \gamma_2$ and the generation of source sequence $x_u$ is not considered as a computation in the embedding.}
\end{displaymath}
The total number of computations of all the functions in $\mathcal{E}$ is $g({\mathcal{E}}) := \sum\limits_{\theta \in \Theta} g_{\mathcal{E}}(\theta).$
 
  Now we will construct a routing-computing scheme with the following properties.
 \begin{enumerate}
  \item It computes $K = d \sum_{\embedding \in \setofembeddings} x'(\embedding)$ realizations of the function with $dx'(\embedding)$ realizations computed by embedding $\embedding.$
  \item It uses any edge $e$ to communicate $N_e = d \sum\limits_{\embedding \in \mathbb{E}} r_{\embedding}(e) x'(\embedding)$ bits, where $r_{\embedding}(e) = \sum\limits_{\theta \in \Theta} r_{\embedding}^{\theta}(e) w(\theta).$
  \item It has $L = d \sum_{\mathcal{E} \in \mathbb{E}} L(\mathcal{E})x'(\mathcal{E}) + d \sum_{\mathcal{E} \in \mathbb{E}} g(\mathcal{E})x'(\mathcal{E}) $ events out of which the number of communication events is $d \sum_{\mathcal{E} \in \mathbb{E}} L(\mathcal{E})x'(\mathcal{E})$ and $d \sum_{\mathcal{E} \in \mathbb{E}} g(\mathcal{E})x'(\mathcal{E}) $ are the computation events.
 \end{enumerate}

Note that for this routing-computing scheme $N_e(\sum_{\embedding \in \setofembeddings} x(\embedding)-\epsilon) \leq N_e \sum_{\embedding \in \setofembeddings} x'(\embedding).$ As $x'(\embedding)$ is a solution of the CALP it satisfies the capacity constraints thus 
\begin{displaymath}
 \sum\limits_{\embedding \in \setofembeddings} r_{\embedding}(e) x'(\embedding)\leq c(e) ~\forall e \in \netedges.
\end{displaymath} 
Using the values of $N_e$ and $K$ for this scheme we get, $ N_e = d \sum\limits_{\embedding \in \setofembeddings} r_{\embedding}(e) x'(\embedding) \leq d c(e) \leq \frac{K c(e)}{\sum\limits_{\embedding \in \setofembeddings} x'(\embedding)}.$
Thus the routing-computing scheme satisfies $N_e(\sum_{\mathcal{E} \in \mathbb{E}} x(\mathcal{E})-\epsilon) \leq N_e \sum_{\mathcal{E} \in \mathbb{E}} x'(\mathcal{E}) \leq Kc(e), ~\forall e \in \netedges.$ This guarantees the achievability of the computing rate $\sum_{\embedding \in \setofembeddings} x(\embedding).$ We now show the sequencing of communication and computation events in the routing-computing scheme.

For this we first compute a total ordering $\tau$ on the vertices and edges of the computation DAG using the underlying DAG ordering. Using this ordering one can inductively order the vertices and edges of the network graph $\net$ which are used in an embedding $\embedding.$ Note that every vertex and edge of $\net$ used in $\embedding$ has a function $\theta$ associated with it and the total number of edges (for transmission) and vertices (for computation) used by it are $L(\embedding) + g(\embedding).$ We denote the ordering (and the corresponding function) generated by an embedding $\embedding$ by
\begin{displaymath}
 \phi_{\embedding} : [1:L(\embedding) + g(\embedding)] \mapsto (\netnodes \times \Theta) \cup (\netedges \times \Theta).
\end{displaymath}

Now we find the total number of times a function $\theta$ being used or transmitted by a vertex $u$ in the network in an embedding $\embedding$ as follows.
\begin{align*}
 m_{u}^{\theta}(\embedding) &= \sum_{v \in \netnodes} \mathbbm{1} \{\phi_{\embedding}(l) = ((u,v),\theta)\} + \sum_{\eta \in \sucfunction{\theta}} \mathbbm{1} \{\phi_{\embedding}(l) = (u,\eta)\}
\end{align*}

We define the sets $\mathcal{U}_{u,l} \subseteq \mathcal{U}; ~\forall u \in \netnodes$ and $~\forall l \in [1,L+1]$ below in an inductive fashion.

\begin{enumerate}
 \item For $1\leq i\leq \kappa,$ $\mathcal{U}_{s_i,1} = \{(\theta_i,k) | k \in [1,K]\}.$ And $\mathcal{U}_{u,1} = \emptyset$ for all $u \in \netnodes \setminus \{s_i|1 \leq i \leq \kappa\}.$
 \item Let us fix an arbitrary order on the embeddings, say $\mathcal{E}_1,\mathcal{E}_2,\ldots,\mathcal{E}_{|\mathbb{E}|}.$ Recall that the $i$-th embedding generates $dx'(\embedding_i)$ number of function symbols. We describe the procedure for the $j$-th symbol generated by $i$-th embedding. The same procedure is run for each symbol of every embedding by following the order of embeddings. 
 Set $m_{u,j}^{\theta} = m_{u}^{\theta}(\embedding_i)$ for all $\theta \in \Theta.$ The scheme for this $j$-th symbol produced by $i$-th embedding has $L(\embedding_i) + g(\embedding_i)$ number of events. We give the procedure for the $l$-th event of this symbol inductively by assuming that all the events till the generation of $(j-1)$-th symbol by $\embedding_i$ and $(l-1)$-th event of $j$-th symbol are right. Then at the $l$-th event do one of the following.
 
   \begin{enumerate}
    \item If $\phi_{\embedding_i}(l) =(u,\theta),$ then the $l$-th event is a computation of $\theta$ at $u.$ The condition $\prefunction{\theta} \subseteq \mathcal{U}_{u,l}(k)$ holds because of the assumption of the correctness of the earlier steps. We set $m_{u,k}^{\eta} = m_{u,k}^{\eta} -1 ~\forall \eta \in \prefunction {\theta}$ and $Z(\mathcal{U}_{u,l}) := \{(\gamma,k) \in \mathcal{U}_{u,l} | m_{u,k}^{\gamma} = 0\}.$ The data-sets are redefined as follows:    $    \mathcal{U}_{u,l+1} = \{\theta,k\} \cup \mathcal{U}_{u,l} \setminus Z(\mathcal{U}_{u,l}) ,$ $\mathcal{U}_{v,l+1} = \mathcal{U}_{v,l}, ~\forall v \in \netnodes \setminus \{u\}.$  Note that this is in accordance with the condition $2(a)$ of Definition~\ref{def:routing-scheme}.
  \item If $\psi_{\embedding_i}(l) = ((u,v),\theta),$ then the $l$-th event is a communication of $\theta(\mathbb{X}(k))$ from $u$ to $v$ over the edge $(u,v).$ $(\phi_{\mathcal{E}_i}(n),k) \subseteq \mathcal{U}_{u,l}(k)$ holds because of the assumption. We first set $m_{u,k}^{\theta} = m_{u,k}^{\theta} -1$ and $Z(\mathcal{U}_{u,l}) := \{(\gamma,k) \in \mathcal{U}_{u,l} | m_{u,k}^{\gamma} = 0\}.$ The redefine the data-sets as follows: $\mathcal{U}_{u,l+1} = \mathcal{U}_{u,l} \setminus Z(\mathcal{U}_{u,l}),$ and $\mathcal{U}_{v,l+1} = \mathcal{U}_{v,l} \cup \{(\theta,k)\}.$ For any $w \neq u,v,$ $\mathcal{U}_{w,l+1} = \mathcal{U}_{w,l}.$ Note that this is in accordance with the condition $2(b)$ of Definition~\ref{def:routing-scheme}.
  \end{enumerate}
\end{enumerate} 
It is easy to verify by running the above procedure inductively the final conditions, $\mathcal{U}_{t,L+1} = \{(f,k)| 1 \leq k \leq K\},
     \mathcal{U}_{u,L+1} = \emptyset ~\forall u \neq t$ and $m_{u,k}^{\theta} =0 ~\forall u,k,\theta$ are met. Similarly the link usage $N_e = \sum_{\theta}r_e^{\theta}w(\theta)$ for all $e \in \netedges$ is also satisfied, where
     \begin{displaymath}
       r_e^{\theta} = |\{l \in [1,L]: l \mbox{ is a communication over }e \mbox{ for function } \theta\}|.
     \end{displaymath}

 \textbf{Step $2$ of the proof:} Now we prove that for any $(\{N_e|e \in \netedges\},K,m_{u,k}^{\theta})$ routing-computing scheme for $(\net,\compgraph)$ satisfying $N_e\lambda \leq Kc(e), ~\forall e \in \netedges$ there exists $\{x(\mathcal{E})|\mathcal{E} \in \mathbb{E}\}$ satisfying the constraints of CALP such that $\sum\limits_{\mathcal{E} \in \mathbb{E}} x(\mathcal{E}) = \lambda.$ 
 
 In any routing-computing scheme looking at the communication and computation events corresponding to the $k$-th symbol of all the functions one can easily get an embedding. Let us say that for the $k$-th computation the scheme uses embedding $\mathcal{E}^{(k)} \in \mathbb{E}.$ For each $e \in \netedges,$ the $k$-th computation requires communication of $r_{\mathcal{E}^{(k)}}^{\theta}(e)$ bits over $e$ of function type $\theta.$ Usage of the link $e$ by the embedding $\embedding^{(k)}$ can be computed by $r_{\mathcal{E}^{(k)}}(e) = \sum\limits_{\theta \in \Theta} r_{\mathcal{E}^{(k)}}^{\theta}(e) w(\theta).$ Thus the total link usage by the scheme can be written as
 \begin{equation}
  \sum_{k=1}^{K} r_{\mathcal{E}^{(k)}}(e) = N_e  ~\forall e \in \netedges. \label{eq:proof_ne}
 \end{equation}
 Let $x(\mathcal{E}) := \frac{\lambda |k \in [1,K]: \mathcal{E}^{(k)}  \in \mathbb{E}|}{K} ~\forall \mathcal{E} \in \mathbb{E}.$
 Note that by definition, $x(\embedding) \geq 0$ and   $\sum\limits_{\mathcal{E} \in \mathbb{E}} x(\mathcal{E}) = \lambda.$ Equation~\eqref{eq:proof_ne} can be written as
 \begin{eqnarray*}
  \sum\limits_{\mathcal{E} \in \mathbb{E}} |k \in [1,K]: \mathcal{E}^{(k)} = \embedding|r_{\mathcal{E}}(e) = N_e \\
  \sum\limits_{\mathcal{E} \in \mathbb{E}} K x(\mathcal{E}) r_{\embedding}(e) = \lambda N_e \leq K c(e) \\
   \sum\limits_{\mathcal{E} \in \mathbb{E}} x(\mathcal{E}) r_{\embedding}(e) \leq c(e)
 \end{eqnarray*}
So, $\{x(\mathcal{E})|\mathcal{E} \in \mathbb{E}\}$ satisfies the conditions of the CALP. Thus we get a solution of CALP with $\sum\limits_{\mathcal{E} \in \mathbb{E}} x(\mathcal{E}) = \lambda$ from the routing-computing scheme.

 
\end{proof}

\section{Complexity of CALP}
\label{sec:hardness}

In this section we prove that solving CALP is MAX SNP-hard even when $\compgraph$ has bounded degree and bounded edge weights. We first prove that if there is an $\alpha$-approximation for CALP then there is an $\alpha$-approximation algorithm for \emph{minimum cost embedding} problem. We give a linear reduction from \maxcut\ to the problem of finding \emph{minimum cost embedding.} Because \maxcut\ is a MAX SNP-hard problem, we get the following theorem.

\begin{theorem}
 \label{thm:main}
 For a DAG $\compgraph$ and arbitrary $\net$ solving CALP is MAX SNP-hard even when: (1)~Each vertex of $\compgraph$ (except for the sink) has bounded ($O(1))$ degree. (2)~Every edge of $\compgraph$ has bounded $(O(1))$ weight. (3)~All the outgoing edges of a vertex of $\compgraph$ have same weight. (4)~The network graph $\net$ has only three vertices.
\end{theorem}

\textbf{Proof Outline:} We give the reduction in several steps. The outline of the proof is as follows.
\begin{enumerate}
  \item We first consider the dual of CALP and its separation oracle which is  a version of the problem of finding the \emph{minimum cost embedding.}
  \item We then prove that there is an $\alpha$-approximation for CALP if and only if there is an $\alpha$-approximation for the separation oracle of its dual. This implies that if \emph{minimum cost embedding} problem is hard to approximate beyond some factor then finding the maximum rate of computation is also hard to approximate.  
 \item Next we prove MAX SNP-hardness of by reducing \maxcut\ problem to \emph{minimum cost embedding}. We use a series of gadgets to obtain the desired properties of the computation graph $\compgraph.$
 
\end{enumerate}
\subsection{Step 1 of the proof}
\label{sec:step1}

First we consider the dual of CALP which is presented below. Recall that $\setofembeddings$ represents the set of all possible embeddings of $\compgraph$ on $\net$ and $r_{\embedding}(e)$ represents the number of times an edge $e \in \netedges$ is used by the embedding $\embedding.$

\centerline{\rule{\columnwidth}{0.75pt}}
\textbf{Dual of CALP}

\textbf{Objective:} Minimize $C = \sum_{e \in \netedges} c(e) y(e)$
\textbf{subject to}
\begin{enumerate}
 \item Cost constraints: $ \sum\limits_{e \in \netedges} r_{\embedding}(e)y(e) \geq 1, \mbox{ } ~\forall \embedding \in \setofembeddings,$ where $r_{\embedding}(e)= \sum_{\theta \in \Theta} r_{\embedding}^\theta(e) w(\theta)$.
 \item Non-negativity constraints: $ y(e) \geq 0 ~\forall e \in \netedges.$
\end{enumerate}

\centerline{\rule{\columnwidth}{0.75pt}}

Note that $r_{\embedding}(e)$ can be computed given the embedding $\embedding.$ Given a vector $\{x(e)| e \in \netedges\}$ the total cost of an embedding can be defined as:
\begin{equation}
 C(\embedding) := \sum_{e \in \netedges} r_{\embedding}(e) x(e) = \sum_{e \in \netedges} \left(\sum\limits_{\theta \in \Theta} r_{\embedding}^\theta(e) w(\theta)\right) x(e). \label{eq:orig_cost}
\end{equation}

Observe that for any given solution of the dual of CALP, $\{y(e)| e \in \netedges\},$ a cost constraint corresponding to an embedding $\embedding$ is $C(\embedding) \geq 1.$ Let us now look at the separation oracle of the dual of CALP.

\begin{definition}[\textbf{Separation oracle of Dual of CALP}]
\textbf{Instance:} A network graph $\net,$ a computation DAG $\compgraph,$ weight function $\{w(\theta) | \theta \in \Theta\}$ and a vector $\{y(e) | e \in \netedges\}.$
\textbf{Output:} If $C(\embedding) \geq 1 ~\forall \embedding \in \setofembeddings,$ then output ``yes" else output ``no" and an embedding $\embedding$ such that $C(\embedding) <1.$
\end{definition}

Note that to solve the above problem, it suffices to compute the minimum cost embedding of $\compgraph$ on $\net.$ A version of minimum cost embedding problem has been studied in \cite{Vyavahare14}. We formally define this cost in Section~\ref{sec:app_algo} and then derive its relation to the cost defined in Equation~\eqref{eq:orig_cost}. In the next section we prove the relation between CALP and the problem of finding minimum cost embedding of $\compgraph$ on $\net.$

\subsection{Step 2 of the proof}
\label{sec:step1.1}

In this section we prove the equivalence between the the problem of solving CALP and the separation oracle of its dual, which is to find the minimum cost embedding. In the process we present a procedure to find a solution of CALP if we have an algorithm to solve minimum cost embedding problem. This will be used in Section~\ref{sec:app_algo} to approximately solve CALP. Specifically we prove the following theorem.
\begin{theorem}
 \label{thm:approx_separation}
 There is a polynomial time $\alpha$-approximation algorithm to solve CALP if and only if there is a polynomial time $\alpha$-approximation algorithm for finding the minimum cost embedding of $\compgraph$ on $\net.$
\end{theorem}
\begin{proof}
The arguments to prove the theorem are similar to the one presented in Theorem~4 of \cite{Jain03} where they consider a packing Steiner tree LP. The main difference between their packing LP and our LP is that in their case the coefficient of the dual variables $\{y(e) | e \in \netedges\}$ are $0/1.$ In our LP (the dual of CALP) the coefficient is $r_{\embedding'}(e)$ which could be any positive number depending on the embedding $\embedding'.$

\textbf{In the forward direction} starting from an $\alpha$-approximation polynomial time algorithm, say $A,$ for the minimum cost embedding we give an $\alpha$-approximation polynomial time algorithm to solve the CALP. First we add the inequality $\sum_{e \in \netedges} c(e)y(e) \leq R$ in the constraints of dual of CALP and using ellipsoid algorithm and binary search (over various values of $R$) we find the minimum value of $R$, say $R^*$, for which the dual is feasible. We use the algorithm $A$ for the separation oracle of dual while running the ellipsoid method. The separation oracle works as follows: First for a given set of $\{y(e)\}$ it checks the inequality $\sum_{e \in \netedges} c(e)y(e) \leq R.$ If this is true then it uses algorithm $A$ to find the minimum cost embedding $\embedding$ of cost $C(\embedding).$ If $C(\embedding) <1$ then we know that $\{y(e)\}$ is not a feasible solution of the dual and $\embedding$ gives a separating hyperplane. But if $C(\embedding) >1$ then $\{y(e)\}$ is considered to be a feasible solution and the corresponding dual (with the added inequality) is considered feasible. Since algorithm $A$ is an $\alpha$-approximation of the optimal minimum cost embedding we know that the above conclusion might be incorrect and the dual might indeed be infeasible. However, in this case $\{\alpha y(e)\}$ gives the feasible solution with $R$ replaced by $\alpha R.$ Note that, this is possible because the right hand side of the cost constraints is all $1$ in the dual. Therefore if $R^{*}$ is the minimum value of $R$ found feasible by the ellipsoid method then we know that the optimal solution of dual lies in the range between $R^{*}$ and $\alpha R^{*}.$ Thus by strong duality of linear programming this method gives us $\alpha$ approximation value of the solution of CALP. 

To find the actual solution corresponding to this value, i.e., $\{x(\embedding) \forall \embedding \in \setofembeddings'\}$ we do the following: We know that the ellipsoid method ends in polynomial time giving polynomially many separating hyperplanes to reach to the $\alpha$-approximate solution. These hyperplanes are sufficient to show that the solution of dual is atleast $R^{*}.$ Corresponding to each of these hyperplanes in the dual there is a variable in the primal CALP. If we set all the other variables to zero then we get a polynomial sized version of CALP whose solution is at least $R^{*}.$ This version of CALP can be solved in polynomial time giving the $\alpha$-approximate solution $\{x(\embedding)\}$ of CALP. This completes the forward direction of Theorem~\ref{thm:approx_separation}.

\textbf{In the reverse direction} we start with an $\alpha$-approximate solution, say $\{x(\embedding)\},$ of CALP and find an $\alpha$-approximate minimum cost embedding. Recall that the objective function value corresponding to this is $x_{sol} = \sum_{\embedding \in \setofembeddings'}x(\embedding).$ By LP-duality we know that $x_{sol} /\alpha$ is an $\alpha$-approximate value of the optimal of dual of CALP and $x_{sol}/\alpha = \sum_{e \in \netedges} c(e) y(e).$ We set each $y(e) := \frac{x_{sol}}{\alpha c(e) |E|}$ to get the corresponding solution (possibly infeasible) of the dual of CALP. 

If $P$ is the polytope defined by the constraints of dual of CALP then we define its polar by $P^{*} := \{z | \langle z,y\rangle \geq 1, \forall y \in P\}.$ It is easy to observe that if we can find an approximate solution over $P$ then we can approximately solve the separation oracle problem of $P^{*}$ and $(P^{*})^{*} =P.$ Using the $\alpha$-approximate solution $\{y(e)\}$ found above we get $\alpha$-approximate separation oracle of $P^{*}.$ Using the ellipsoid method mentioned in the forward direction of the proof and this separation oracle we get an $\alpha$-approximate solution on $P^{*}.$
As $(P^{*})^* =P$  this solution over $P^{*}$ gives an $\alpha$-approximate separation oracle of $P$ which is equivalent to approximately solving the minimum cost embedding problem. In this case also as the right hand side of the edge constraints are all $1,$ the approximation ratio is preserved. 

\end{proof}

\subsection{Step 3 of the proof}
\label{sec:step2}
In Section~\ref{sec:step1.1} we showed that solving CALP is equivalent to solving minimum cost embedding. In this section we reduce a known NP-complete problem, \maxcut\  \cite{Garey79}, to the minimum cost embedding problem thus proving that solving CALP is NP-complete.

A \maxcut\ problem is defined as follows: Given an unweighted graph $H=(V_{H},E_{H})$ and a number $K,$ check whether there is a partition of $V_H$ into two sets $V_1$ and $V_2$ such that there are at least $K$ edges between $V_1$ and $V_2.$ Moreover, it is known that if the input graph of \maxcut\ problem is a cubic graph \footnote{A graph in which each vertex has exactly degree three is called cubic graph.} then the problem is MAX SNP-hard \cite{Berman99}. We start with an instance of \maxcut\ with cubic graph and prove the MAX SNP-hardness of minimum cost embedding problem.

 Given an instance $\phi =\{H,K\}$ of \maxcut\ where $H$ is a cubic graph, we generate an instance of minimum cost embedding problem $\psi = (\compgraph,S_{\compgraph},\omega_p,w;\net,S_\net,t,y).$ Recall that $\net=(\netnodes,\netedges)$ is the network graph with $S_{\net} \subset \netnodes$ sources, $t$ as the sink and $y$ as the weight function on $\netedges.$ Similarly, $\compgraph=(\compnodes,\compedges)$ is a computation DAG with $S_{\compgraph}$ as sources, $\omega_p$ as the sink and $w$ as the weight function on $\compedges.$ 
 
\begin{theorem}
 \label{thm:flow-to-cost}
 For an instance $\phi$ of \maxcut\ we construct an instance $\psi$ of the minimum cost embedding such that $\phi$ has a cut of size at least $K$ if and only if $\psi$ has an optimal embedding of cost at most $28|E_{H}| -K.$

\end{theorem} 

\begin{proof}
First we create an undirected network graph. We consider $\net$ to be a complete graph on three vertices with $V = \{S_1,S_2,t\}.$ We set $S_{\net} = \{S_1,S_2\}$ as the sources and $t$ as the sink vertex. We set the weight $y(e) =1 ~\forall e \in \netedges.$  

Now we create the computation graph $\compgraph$ from $H$ using a series of gadgets each of which enables the desired properties on $\compgraph$ as follows: We start with the gadget shown in Fig.~\ref{fig:diamond_gadget}(a) for each edge $(x,y) \in E_{H}.$ This gadget is used to prove the MAX SNP-hardness of Multiterminal cut from \maxcut\ in \cite{Dahlhaus94}. We direct readers to \cite{Dahlhaus94} for more details of this gadget. Note that each $S_{ixy}$ is connected to four vertices with edges of weight four. We create four vertices for each $S_{ixy}$ (one for each of its one outgoing edge) and connect one of its neighbor of $S_{ixy}$ to exactly one of these newly created vertices. We put the directions on the edges of Fig.~\ref{fig:diamond_gadget}(a) such that all the edges from $S_{ixy}$ are outgoing edges and $\omega_p$ has all incoming edges. The resulting gadget is shown in Fig.~\ref{fig:diamond_gadget}(b). It is easy to observe that Fig.~\ref{fig:diamond_gadget}(b) is just a redrawn directed version of Fig.~\ref{fig:diamond_gadget}(a) with a separate vertex for each edge of $S_{ixy}.$ We denote the graph formed by replacing each edge of $E_H$ by the gadget of Fig.~\ref{fig:diamond_gadget}(b) by $I.$ Finally, we replace every vertex of $I,$ with multiple outgoing edges, by the gadget shown in Fig.~\ref{fig:edge_gadget}. 
 
\begin{figure}[t]
\begin{center}
 \subfloat[]{
\begin{tikzpicture}[>=latex]
\footnotesize
 \foreach \x in {-2,0,2}{
  \foreach \y in {-2,0,2}{
    \fill (\x,\y) circle (0.07cm);
   }
   }
 
  \foreach \x in {-2,0}{
   \foreach \y in {-2,2}{
     \draw [-] (\x,\y) -- (\x+2,\y);
   }
   }
  
  \foreach \y in {2,0}{
   \foreach \x in {-2,2}{
     \draw [-] (\x,\y) -- (\x,\y-2);
   } 
   }
   \foreach \x in {-2,2}{
     \draw [-] (0,0) -- (\x,0) node[draw=none,midway,below] {$4$};
    \draw [-] (0,0) -- (0,\x) node[draw=none,midway,left] {$4$};
    }
    
   \node at (-1,1.7) [draw=none] {$4$};
   \node at (1,-1.7) [draw=none] {$4$};
   
   \node at (-1.7,1) [draw=none] {$4$};
   \node at (1.7,-1) [draw=none] {$4$};
   
    \draw[-]  (-2,2) to[out=30,in=150] node [draw=none,pos=0.5,above] {$4$} (2,2);
    \draw[-]  (-2,2) to[out=240,in=120] node [draw=none,pos=0.5,left] {$4$} (-2,-2);
     \draw[-]  (-2,-2) to[out=330,in=210] node [draw=none,pos=0.5,below] {$4$} (2,-2);
      \draw[-]  (2,2) to[out=300,in=60] node [draw=none,pos=0.5,right] {$4$} (2,-2);
       \draw[-]  (-2,0) to[out=30,in=150] node [draw=none,pos=0.75,above] {} (2,0);
     \draw[-]  (0,2) to[out=300,in=60] node [draw=none,pos=0.75,right] {} (0,-2);   
     
   \node at (-2,2.2) [draw=none,left] {$S_{1xy}$};
   \node at (0,2.2) [draw=none] {$x$};
   \node at (2,2.2)[draw=none,right] {$a_{xy}$};
   
   \node at (-2,0.2) [draw=none,left] {$y$};
   \node at (-0.3,0.2) [draw=none] {$S_{2xy}$};
   \node at (2,0.2)[draw=none,right] {$b_{xy}$};
   
   \node at (-2,-2.2) [draw=none,left] {$d_{xy}$};
   \node at (0,-2.2) [draw=none] {$c_{xy}$};
   \node at (2,-2.2)[draw=none,right] {$\omega_p$};
\end{tikzpicture}
}
\subfloat[]{
\begin{tikzpicture}[>=latex]
\footnotesize
 \foreach \x in {-2,-1,1,2}{
  \foreach \y in {2,-1}{
    \fill (\x,\y) circle (0.07cm);
   }
   }
 
  \foreach \x in {-1.5,1.5}{
    \fill (\x,1) circle (0.07cm);
   }
   \foreach \x in {-2.5,-0.5,0.5,2.5}{
    \fill (\x,0) circle (0.07cm);
   } 
  \fill (0,-2.5) circle (0.07cm);
  
  \foreach \x in {-2,-1}{
     \draw [-triangle 45] (\x,2) -- (-1.5,1) node[draw=none,midway,right] {$4$};
     \draw [-triangle 45] (-1.5,1) -- (\x,-1);
    } 
    
  \foreach \x in {2,1}{
     \draw [-triangle 45] (\x,2) -- (1.5,1) node[draw=none,midway,right] {$4$};
     \draw [-triangle 45] (1.5,1) -- (\x,-1);
    } 
 
 \foreach \x in {-2,-1,1,2} {
  \draw [-triangle 45] (\x,-1) -- (0,-2.5) node[draw=none,midway,right] {$4$};
  }
  
  \draw [-triangle 45] (-2.5,0) -- (-2,-1) node[draw=none,midway,left] {$4$};
  \draw [-triangle 45] (-0.5,0) -- (-1,-1)node[draw=none,midway,right] {$4$};
  \draw [-triangle 45] (0.5,0) -- (1,-1)node[draw=none,midway,left] {$4$};
  \draw [-triangle 45] (2.5,0) -- (2,-1)node[draw=none,midway,right] {$4$};
  
  \draw [-triangle 45] (-2,-1) to[out=330,in=210] (1,-1);
  \draw [-triangle 45] (2,-1) to[out=210,in=330] (-1,-1);
  \node at (-1.7,1) [draw=none] {$x$};
   \node at (-2.4,-1) [draw=none] {$a_{xy}$};
   
   \node at (1.3,1) [draw=none] {$y$};
   \node at (0.5,-1)[draw=none] {$b_{xy}$};
   
   \node at (2.3,-1) [draw=none] {$d_{xy}$};
   \node at (-1.3,-1) [draw=none] {$c_{xy}$};
   \node at (0,-2.7)[draw=none] {$\omega_p$};
   
   \node at (-2,2.3) [draw=none] {$S_{1xy}^{x}$};
   \node at (-1,2.3) [draw=none] {$S_{2xy}^{x}$};
   \node at (1,2.3) [draw=none] {$S_{1xy}^{y}$};
   \node at (2,2.3) [draw=none] {$S_{2xy}^{y}$};
   \node at (-2.5,0.3) [draw=none] {$S_{1xy}^{a}$};
   \node at (-0.5,0.3) [draw=none] {$S_{2xy}^{c}$};
   \node at (0.5,0.3) [draw=none] {$S_{2xy}^{b}$};
   \node at (2.5,0.3) [draw=none] {$S_{1xy}^{d}$};
\end{tikzpicture}
}
\end{center}
 \caption[Gadget for edges in $H$]{(a) Gadget for edge $(x,y)$ in $H.$ (b) Redrawing the gadget shown in (a) with new vertices for each outgoing edge of $S_{ixy}.$ Numbers near the edges represent their weights and the unlabeled edges have weight $1.$}
  \label{fig:diamond_gadget}
\end{figure}
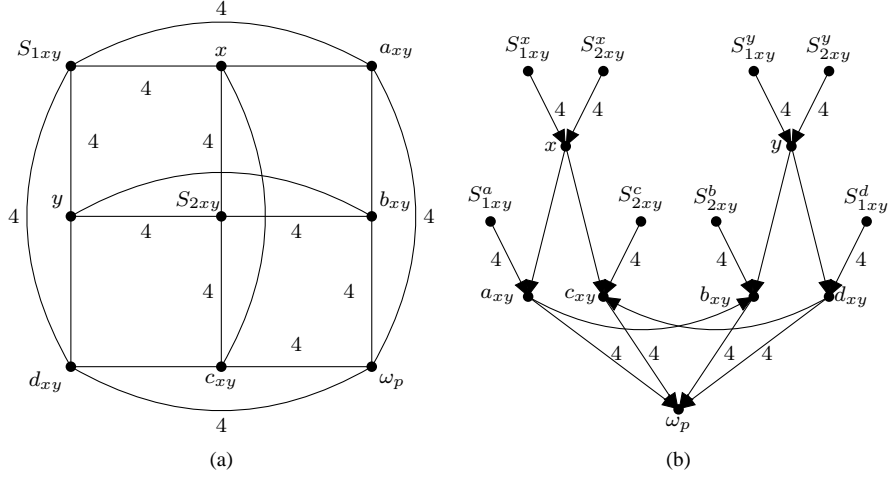  

We set all the vertices of type $S_{ixy}^{*}$ as sources, i.e., $S_{\compgraph} = \{S_{ixy}^{*} | x,y \in V_{H}, i\in \{1,2\}, * \in \{x,y,a,b,c,d\}\}$ and the sink is $\omega_p.$ From each edge gadget we get eight sources thus $|S_{\compgraph}| = 8 |E_H|.$  Similarly, the sink vertex $\omega_p$ has $4 |E_{H}|$ incoming edges. Observe that graph $\compgraph$ has the following properties.

\begin{lemma}
 \label{lm:valid_dag}
 The DAG $\compgraph$ created from an instance $\phi$ of \maxcut\ has the following properties: (1)~All the vertices in $S_{\compgraph}$ have only outgoing edge and  the sink vertex $\omega_p$ has only incoming edges. (2)~All the intermediate vertices in $\compgraph$ have atleast one incoming and one outgoing edge. (3)~There are no directed cycles in $\compgraph.$ (4)~Out-degree of each vertex is bounded. (5)~Weight on each edge is bounded.
\end{lemma}
\begin{proof}
The proof directly follows from the gadgets. Details of the proof are presented in Appendix~\ref{app:lemmas}.
\end{proof}

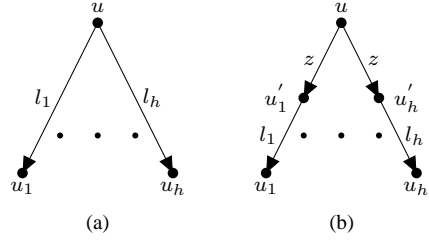
\begin{figure}[t]
 \begin{center}
   \subfloat[]{ 
 \begin{tikzpicture}[>=latex]
\footnotesize
 \fill (0,1) circle (0.07cm);
 \fill (-1,-1) circle (0.07cm);
 \fill (1,-1) circle (0.07cm);
 \draw [-triangle 45] (0,1) -- (-1,-1) node [midway,left] {$l_1$};
 \draw [-triangle 45] (0,1) -- (1,-1) node [midway,right] {$l_h$};
 \node at (0,1.2) [draw=none] {$u$};
 \node at (-1,-1.2) [draw=none] {$u_1$};
 \node at (1,-1.2) [draw=none] {$u_h$};
 
 \foreach \x in {-0.5,0,0.5} {
   \fill (\x,-0.5) circle (0.04cm);
  }
  
 \end{tikzpicture}
}
\hspace{10pt}
\subfloat[]{
\begin{tikzpicture}[>=latex]
\footnotesize
 \fill (0,1) circle (0.07cm);
  \fill (-1,-1) circle (0.07cm);
 \fill (1,-1) circle (0.07cm);
 \foreach \x in {-0.5,0.5} {
 \fill (\x,0) circle (0.07cm);
 }
  \draw [-triangle 45] (0,1) -- (-0.5,0) node [midway,left] {$z$};
  \draw [-triangle 45] (0,1) -- (0.5,0) node [midway,right] {$z$};
 
 \draw [-triangle 45] (-0.5,0) -- (-1,-1) node [midway,left] {$l_1$};
 \draw[-triangle 45] (0.5,0) -- (1,-1) node [midway,right] {$l_h$};
 
 \node at (0,1.2) [draw=none] {$u$};
 \node at (-1,-1.2) [draw=none] {$u_1$};
 \node at (1,-1.2) [draw=none] {$u_h$};
 \node at (-0.6,0) [draw=none,left] {$u_1^{'}$};
 \node at (0.6,0) [draw=none,right] {$u_h^{'}$};
 
 \foreach \x in {-0.5,0,0.5} {
   \fill (\x,-0.5) circle (0.04cm);
  }
 \end{tikzpicture}
}
\end{center}
  \caption[Gadget for outgoing edges of $I$]{(a) A vertex in $I$ with $h$ outgoing edges (b) Gadget to replace the vertex shown in (a). Labels near the edges represent their weights and $z = h \mbox{max}(l_1,\ldots,l_h) +1$}
  \label{fig:edge_gadget}
\end{figure}  
 
Recall that the network graph generated from \maxcut\ has only three vertices. We assume that each source vertex of type $S_{1*}^{*}$ in $\compgraph$ is generated at $S_1 \in \netnodes.$ Similarly, each source of type $S_{2*}^{*}$ is generated at $S_2.$ The sink vertex $\omega_p \in \compnodes$ is mapped to $t \in \netnodes.$ This completes the generation of an instance $\psi$ from $\phi$ of \maxcut.

Before we start proving Theorem~\ref{thm:flow-to-cost}, we prove some properties of the gadgets of Figs.~\ref{fig:diamond_gadget},~\ref{fig:edge_gadget}. We say that an edge of $\compgraph$ is exposed in an embedding if its weight is considered while computing the cost of the embedding.

\begin{lemma}
 \label{lm:edgewt}
 In the minimum cost embedding of $\compgraph$ on $\net,$ any edge of weight $z$ is never exposed from the gadget of Fig.~\ref{fig:edge_gadget}(b).
\end{lemma}

\begin{lemma}
\label{lm:mincost_rembedding}
If a $3$- way multiterminal cut (with terminals being $S_{1xy},S_{2xy}, \omega_p$) of the gadget shown in Fig.~\ref{fig:diamond_gadget}(a) has weight $W$ then there is an embedding of the gadget of Fig.~\ref{fig:diamond_gadget}(b) (along with the Fig.~\ref{fig:edge_gadget}(b)) of cost $W$ on $\net.$
\end{lemma} 

Using Lemma~\ref{lm:mincost_rembedding}, we can borrow the following result from Lemma~4.1 of \cite{Dahlhaus94} for the $3$-way cut of Fig.~\ref{fig:diamond_gadget}(a). Refer to \cite{Dahlhaus94} for more details.

\begin{lemma}
 \label{lm:cut_cost_opt}
 There are embeddings of the gadget of Fig.~\ref{fig:diamond_gadget}(b) (along with Fig.~\ref{fig:edge_gadget}(b)) on $\net$ with the following properties.
 \begin{enumerate}
  \item There is an embedding with cost $27$ in which $x,a_{xy}$ are mapped to $S_1;$ $y,b_{xy}$ to $S_2$ and $c_{xy},d_{xy}$ to $t.$ Similarly, there is an embedding with cost $27$ in which $y,d_{xy}$ are mapped to $S_1;$ $x,c_{xy}$ to $S_2$ and $a_{xy},b_{xy}$ to $t.$
  \item Any other embedding in which $x$ is mapped to $S_1$ but $y$ is not mapped to $S_2$ or vice a versa has cost strictly greater than $27.$ 
  \item Moreover, there are embeddings in which $x,y$ both are either mapped to $S_1$ or $S_2$ have cost exactly $28.$ For example, an embedding in which $x,y,a_{xy}$ are mapped to $S_1;$ $b$ to $S_2$ and $c_{xy},d_{xy}$ to $\omega_p$ has cost $28.$ Similarly, an embedding in which $x,y,c_{xy}$ are mapped to $S_2;$ $d_{xy}$ to $S_1$ and $a_{xy},b_{xy}$ to $\omega_p$ has cost $28.$
 \end{enumerate}
\end{lemma}

And finally we need the following lemma to prove Theorem~\ref{thm:flow-to-cost}.
\begin{lemma}
\label{lm:cost_cut}
Given any embedding $\embedding$ with cost $C(\embedding)$ of $\compgraph$ on $\net$ in which a vertex of $\compgraph$ is mapped to multiple vertices of $\net$ we can obtain an embedding $\embedding'$ in which no vertex of $\compgraph$ is mapped to more than one vertex of $\net$ and has cost $C(\embedding') \leq C(\embedding)$ in polynomial time.
\end{lemma}

Proofs of all these lemmas are presented in Appendix~\ref{app:lemmas}.

\textbf{Proof of forward direction (Theorem~\ref{thm:flow-to-cost}):} We need to prove that if there is a \maxcut\ of graph $H$ of size at least $K$ then there is an embedding of cost at most $28|E_{H}| -K$ of $\compgraph$ on $\net.$ Suppose there is a partition of $V_H$ into sets $V_1,V_2$ such that the number of edges between them is at least $K.$ Then we create an embedding of $\compgraph$ on $\net$ as follows: Map all the vertices of $V_1$ to $S_1$ and $V_2$ to $S_2.$ Thus for every edge gadget $x,y$ are either mapped to $S_1$ or $S_2.$ If $x,y$ both are mapped to different $S_i, i \in \{1,2\}$ then map the intermediate vertices of this gadget according to the embedding of Lemma~\ref{lm:cut_cost_opt} point~1 and if they are mapped to the same vertex then use the embeddings given in point~3 of Lemma~\ref{lm:cut_cost_opt}. Specifically, if $x,y$ are in the same set in the \maxcut\ then the gadget will contribute $28$ to the cost of the embedding else it will contribute $27.$ As there are at least $K$ edges across the cut, the total cost of the embedding of $\compgraph$ on $\net$ is at most $28|E_H| -K.$

\textbf{Proof of backward direction (Theorem~\ref{thm:flow-to-cost}):} Now we need to prove that if there is a minimum cost embedding of cost less than $28|E_H| -K$ then there is a cut of size at least $K$ for $H.$ From Lemma~\ref{lm:cost_cut} we know that the minimum cost embedding maps every vertex of $\compgraph$ to only one vertex of $\net.$ For each edge $(x,y) \in E_H$ we know from Lemma~\ref{lm:cut_cost_opt} (point~2) that the cost of the embedding from its gadget is $\geq 28$ unless $x,y$ (or $y,x$) are mapped to $S_1,S_2$ (or $S_2,S_1$) respectively. If the cost of the embedding is less than $28|E_H| -K$ then there must be at least $K$ edge gadgets in which $x,y$ (or $y,x$) are mapped to $S_1,S_2$ (or $S_2,S_1$) respectively. To get a cut of $H$ from this embedding we take $\{x | x \in V_H\}$ which are mapped to $S_1$ to be in $V_1$ and the vertices which are mapped to $S_2$ to be in $V_2.$ The vertices of $V_H$ which are mapped to $\omega_p$ are arbitrarily put in the set $V_1$ or $V_2.$ By our earlier arguments there are at least $K$ edges between $V_1$ and $V_2$ thus giving a cut of size at least $K.$
\end{proof}

We now show that the reduction presented in Theorem~\ref{thm:flow-to-cost} is indeed a linear reduction thus proving the MAX SNP-hardness of the minimum cost embedding problem \cite{Dahlhaus94}. We just showed that an instance $\phi$ of \maxcut\ with optimal value $\mathsf{opt}(\phi)$ can be converted into an instance $\psi$ of minimum cost embedding problem in polynomial time such that $\mathsf{opt}(\psi) \leq 28|E_H| -\mathsf{opt}(\phi).$ Note that for any instance of \maxcut\ problem $\mathsf{opt}(\phi) \geq |E_H| / 2$ \footnote{A simple greedy algorithm can construct such a cut.}. Thus, 
\begin{equation}
 \mathsf{opt}(\psi) \leq \frac{55}{2} |E_H| \leq 55 \mathsf{opt}(\phi). \label{eq:linearreduction_alpha}
\end{equation}

For any solution $y$ of $\psi$ with $\mathsf{cost}(y)= 28|E_H| -K,$ by Lemma~\ref{lm:cost_cut} we can obtain an embedding $y'$ in which every vertex of $\compgraph$ is mapped to only one vertex of $\net$ and has cost at most $28|E_H| -K.$ Let the cost of this new embedding be $\mathsf{cost}(y') = 28|E_H| - K'$ where $K' \geq K.$ By Theorem~\ref{thm:flow-to-cost} we know that we can obtain a solution $x$ of $\phi$ from $y'$ of weight at least $K'.$ Thus, $|\mathsf{cost}(x) - \mathsf{opt}(\phi)| \leq |K' - \mathsf{opt}(\phi)|.$ On the other hand $|\mathsf{cost}(y) - \mathsf{opt}(\psi)| \geq |28|E_H| -K +28|E_H| +\mathsf{opt}(\phi)|.$ As $\mathsf{opt}(\phi) \geq K' \geq K$ we get,
\begin{equation}
 |\mathsf{cost}(x) - \mathsf{opt}(\phi)| \leq |\mathsf{cost}(y) - \mathsf{opt}(\psi)|. \label{eq:linearreduction_beta}
\end{equation}

Equations~\eqref{eq:linearreduction_alpha},~\eqref{eq:linearreduction_beta} prove that the reduction presented in Theorem~\ref{thm:flow-to-cost} is a linear reduction. Authors in \cite{Berman99} showed that for \maxcut\ no algorithm can achieve an approximation ratio of $0.997$ unless P=NP. Combining with the linear reduction factors of Equations~\eqref{eq:linearreduction_alpha},~\eqref{eq:linearreduction_beta} we get the following result.

\begin{corollary}
 \label{cor:ratecost_snphardness}
 For a given DAG $\compgraph$ and network graph $\net$ finding minimum cost embedding is MAX SNP-hard even when $\compgraph$ has bounded out-degree, weights on its edges are bounded, and $\net$ has only three vertices. Moreover, it is hard to approximate above a factor of $0.0178$ unless P=NP.
\end{corollary}

\section{Algorithm for $\net$ with two vertices}
\label{sec:2nodenet_algo}

In Theorem~\ref{thm:flow-to-cost} (Section~\ref{sec:step2}) we proved that finding minimum cost embedding is NP-hard even when there are only three vertices in $\net.$ In this section we present a polynomial time algorithm to find the minimum cost embedding when the network graph has only two vertices. By using the algorithm presented in this section and the technique of Theorem~\ref{thm:approx_separation} we can obtain a rate maximizing schedule for an arbitrary computation graph on a two node network graph in polynomial time. 

For all the discussion in this section we assume that the network graph $\net$ has two vertices $n_1,n_2$ connected via an edge of weight $x(n_1,n_2).$ The computation graph is assumed to be an arbitrary DAG $\compgraph.$ There are $\kappa$ sources in $\compgraph = (\compnodes,\compedges);$ out of which $\kappa_1$ are mapped to $n_1$ and others are mapped to $n_2.$ The sink vertex $\omega_p$ of $\compgraph$ is at node $n_2.$ There is a weight function $\{w(\gamma) | \gamma \in \compedges\}$ \footnote{Recall that the weight of an edge of $\compgraph$ is associated with the sub-function it carries. Thus all outgoing edges of a vertex of $\compgraph$ have same weight.} associated with the edges of $\compgraph.$ The problem is to find the embedding of $\compgraph$ on $\net$ such that the cost of the embedding is minimized. Recall that cost of an embedding is defined by Equation~\eqref{eq:orig_cost}.

To find the minimum cost embedding we first reduce our problem to an instance of \mincut\ which is defined as follows: Given a directed graph $J=(V_J,E_J)$ with weights on edges $\{g(i,j) | (i,j) \in E_J\}$ and two distinct vertices $j_1,j_2 \in V_J,$ find two disjoint subsets $J_1,J_2 \subset V_J$ such that $j_1 \in J_1, j_2 \in J_2$ and the following optimal value is achieved.
\begin{equation}
 \mathsf{opt}(\mincut(j_1,j_2)) := \min\limits_{J_1,J_2 \in V_J} ( \delta (J_1) + \delta(J_2)). \label{eq:opt_mincut}
\end{equation}
For any set $A \subseteq V_J,$ $\delta(A)$ is defined as the sum of weights of all the outgoing edges from $A.$ In other words,
\begin{equation}
 \delta(A) := \sum\limits_{i \in A, j \in V_J \setminus A} g(i,j). \label{eq:delta-of-set}
\end{equation}

We show that \mincut\ problem can be solved in polynomial time and then present an algorithm which converts the optimal solution of \mincut\ to the corresponding instance of minimum cost embedding of $\compgraph$ on $\net.$ 

\begin{lemma}
 \label{lm:mincut-poly}
Given any directed graph $J$ and its two distinct vertices $j_1,j_2$ \mincut\ can be solved in polynomial time.
\end{lemma}
\begin{proof}
 Recall that the solution of $\mathsf{opt}(\mincut(j_1,j2)$ are two disjoint subsets $J_1,J_2$ of $V_J$ such that $j_1 \in J_1$ and $j_2 \in J_2.$ Equation~\eqref{eq:opt_mincut} can be written as $ \mathsf{opt}(\mincut(j_1,j_2)) = \min\limits_{J_1 \in V_J} [ \delta(J_1) + \min\limits_{J_2 \subseteq V_J \setminus J_1} \delta(J_2)].$ For a given $J_1$ we need to compute the right hand side of the above equation in polynomial time. To do so we modify the equation as follows: Let $A$ be a subset of $V_J$ such that $j_1 \notin A.$ Then we rewrite the equation as
 \begin{equation}
  \mathsf{opt}(\mincut(j_1,j_2)) = \min\limits_{A \subset V_J} [ \delta(A \cup j_1) +\min\limits_{C \subseteq V_J \setminus \{A,j_1,j_2\}} \delta(C \cup j_2)]. \label{eq:mincut_delta}
 \end{equation}
The second term of the right hand side of above equation can be computed in polynomial time by computing the minimum cut of $j_2$ by considering the subsets from $V_J \setminus \{A,j_1,j_2\}.$ Thus for a given set $A,$ right hand side of Equation~\eqref{eq:mincut_delta} can be computed in polynomial time. Now we show that this is indeed a submodular function and thus the set $A$ which minimizes the value can also found in polynomial time.

A function $h$ on the subsets of a set $U$ is \emph{submodular} if for any two sets $Y,Z \subseteq U,$ $h(Y) + h(Z) \geq h(Y \cap Z) + h(Y \cup Z).$ 
For any two subsets $Y,Z \subseteq V_J$ it is easy to observe that $\delta(Y \cup Z) \leq \delta(Y) + \delta(Z) - \delta(Y \cap Z).$ Hence $\delta$ is a submodular function. Let $X \subseteq V_J \setminus \{A,j_1,j_2\}$ be the set which minimizes the second term of Equation~\eqref{eq:mincut_delta}. Then for a set $A,$ let $h(A) := \delta(A \cup j_1) + \delta(X \cup j_2).$ Similarly for a set $B$ $h(B) = \delta(B \cup j_1) + \delta(Y \cup j_2)$ where $Y \subseteq V_J \setminus \{B,j_1,j_2\}$ minimizes the second term of Equation~\eqref{eq:mincut_delta}. Also, $ h(A \cup B) = \delta(A \cup B \cup j_1) + \delta (Z \cup j_2)$ for some $Z  \subseteq V_J \setminus \{A \cup B, j_1,j_2\}.$ Note that $(X \cap Y)$ and $(A \cup B)$ are disjoint sets which implies that $X \cap Y \subseteq V_J \setminus \{A \cup B, j_1,j_2\}.$ Thus,
\begin{displaymath}
 h(A \cup B) \leq \delta(A \cup B \cup j_1) + \delta (X \cap Y \cup j_2).
\end{displaymath}
Similarly $h(A \cap B) = \delta(A \cap B \cup j_1) + \delta(W \cup j_2)$ for some $W \subseteq V_J \setminus \{A \cap B, j_1,j_2\}.$ Note that $(X \cup Y)$ and $(A \cap B)$ are disjoint sets. Thus, 
\begin{displaymath}
 h(A \cap B) \leq \delta (A \cap B \cup j_1) + \delta(X \cup Y \cup j_2).
\end{displaymath}
As $\delta$ is a submodular function, it is easy to observe that $h(A \cup B) + h(A \cap B) \leq h(A) + h(B).$ This proves that the right hand side of Equation~\eqref{eq:mincut_delta} is a submodular function and $\mathsf{opt}(\mincut(j_1,j_2))$ can be obtained in polynomial time by using algorithm presented in \cite{Schrijver00}.
\end{proof}

Given an instance $\psi = (\compgraph,S_{\compgraph},\omega_p,w,\net,S_{\net},t,y)$ of minimum cost embedding we create an instance $\phi= (J,g,j_1,j_2)$ of \mincut. 

\begin{theorem}
 \label{thm:cut-to-cost}
 The instance $\psi$ of minimum cost embedding problem has the optimal embedding of cost $C$ if an only if the corresponding instance $\phi$ of \mincut\ has the optimal cut of weight $C.$
\end{theorem}

\begin{proof}
 We first construct the directed graph $J$ for \mincut\ instance from $\compgraph,\net$ as follows: Replace each vertex of $\compgraph,$ except for the sink vertex $\omega_p,$ by the gadget shown in Fig.~\ref{fig:vertex_gadget}. Add two vertices labeled $j_1,j_2$ in this graph. Add outgoing edges from $j_1$ to all the ``in" vertices pf the sources which are mapped to $n_1 \in \net$ with weight of $\infty.$ Similarly add outgoing edges from $j_2$ to the remaining ``in" vertices of the sources and the sink $\omega_p$ (note that these vertices are mapped to $n_2 \in \net)$ with weight $\infty.$  We label the resulting directed graph by $J$ for the \mincut\ instance with $j_1,j_2$ being the two vertices for which $\mathsf{opt}(\mincut(j_1,j_2))$ has to be computed. 
 
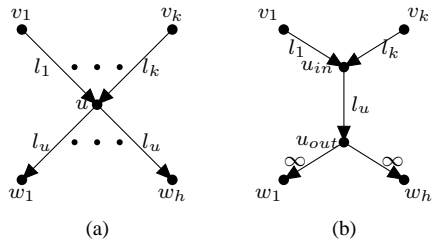
\begin{figure}[t]
 \begin{center}
  \subfloat[]{ 
 \begin{tikzpicture}[>=latex]
\footnotesize
 \fill (0,0) circle (0.07cm);
 \fill (-1,-1) circle (0.07cm);
 \fill (1,-1) circle (0.07cm);
 \fill (-1,1) circle (0.07cm);
 \fill (1,1) circle (0.07cm);
 \draw [-triangle 45] (0,0) -- (-1,-1) node [midway,left] {$l_u$};
 \draw [-triangle 45] (0,0) -- (1,-1) node [midway,right] {$l_u$};
  \draw [-triangle 45] (-1,1) -- (0,0) node [midway,left] {$l_1$};
 \draw [-triangle 45] (1,1) -- (0,0) node [midway,right] {$l_k$};
 \node at (-0.2,0) [draw=none] {$u$};
 \node at (-1,-1.2) [draw=none] {$w_1$};
 \node at (1,-1.2) [draw=none] {$w_h$};
 \node at (-1,1.2) [draw=none] {$v_1$};
 \node at (1,1.2) [draw=none] {$v_k$};
 
 \foreach \x in {-0.3,0,0.3} {
   \fill (\x,-0.5) circle (0.04cm);
   \fill (\x,0.5) circle (0.04cm);
  }
  
 \end{tikzpicture}
}
\hspace{10pt}
\subfloat[]{ 
 \begin{tikzpicture}[>=latex]
\footnotesize
 \fill (0,0.5) circle (0.07cm);
 \fill (-0.8,-1) circle (0.07cm);
 \fill (0.8,-1) circle (0.07cm);
 \fill (0,-0.5) circle (0.07cm);
 \fill (-0.8,1) circle (0.07cm);
 \fill (0.8,1) circle (0.07cm);
 \draw [-triangle 45] (0,-0.5) -- (-0.8,-1) node [midway,left] {$\infty$};
 \draw [-triangle 45] (0,-0.5) -- (0.8,-1) node [midway,right] {$\infty$};
  \draw [-triangle 45] (-0.8,1) -- (0,0.5) node [midway,left] {$l_1$};
 \draw [-triangle 45] (0.8,1) -- (0,0.5) node [midway,right] {$l_k$};
 \draw [-triangle 45] (0,0.5) -- (0,-0.5) node [midway,right] {$l_u$};
 \node at (-0.35,0.5) [draw=none] {$u_{in}$};
 \node at (-0.35,-0.5) [draw=none] {$u_{out}$};
 \node at (-1,-1.2) [draw=none] {$w_1$};
 \node at (1,-1.2) [draw=none] {$w_h$};
 \node at (-1,1.2) [draw=none] {$v_1$};
 \node at (1,1.2) [draw=none] {$v_k$};
  \end{tikzpicture}
}
 \end{center}
 \caption[Gadget for a vertex in $\compgraph$]{(a) A vertex in $\compgraph$ with $k$ incoming edges and $h$ outgoing edges of weight $l_u$ (b) Gadget to replace $u$ shown in (a). Labels near the edges represent their weights.}
  \label{fig:vertex_gadget}
\end{figure}  

Proof of Theorem~\ref{thm:cut-to-cost} follows directly from the following two lemmas.

\begin{lemma}
 \label{lm:embedding-to-2cut}
 If for the instance $\psi$ there is an embedding $\embedding$ of cost $C$ then there is a $\mincut(j_1,j_2)$ of weight $C$ for the instance $\phi.$
\end{lemma}
\begin{proof}
Before proving the lemma we recall a few notations and ideas about $\compgraph$ and its embedding on $\net.$ Every vertex $u$ of $\compgraph$ computes a specific function $\theta$ and all its outgoing edges carry the same function. The set of all the successor functions of $\theta$ is represented by $\sucfunction \theta.$ An embedding of $\compgraph$ on $\net$ gives us a mapping of vertices of $\compgraph$ to that of $\net.$ It tells us on which vertices of $\net$ the function $\theta$ is computed. The network graph $\net$ for our instance $\psi$ has only two vertices $n_1,n_2.$ Thus any function is either computed at $n_1$ or $n_2$ or both. Also recall that the $\mincut(j_1,j_2)$ partitions the vertex set $V_J$ of $J$ into three disjoint sets $J_1,J_2,J_3$ such that $j_1 \in J_1, j_2 \in J_2.$ In all the discussion below we assume that vertex $u \in \compnodes$ computes function $\theta.$ We compute the $\mincut(j_1,j_2)$ from the embedding $\embedding$ as follows:
\begin{enumerate}
 \item Put $j_1$ ($j_2$) in $J_1$ ($J_2$ respectively).
\item If a source vertex $\omega_i \in \compgraph$ is mapped to $n_1$ ($n_2$) then put $\omega_i^{in}$ in $J_1$ ($J_2$ respectively). Put the sink vertex $\omega_p$ in $J_2.$
 \item If $\theta$ is computed at both $n_1,n_2$ under embedding $\embedding$ then put $u^{in},u^{out}$ in $J_3.$
 \item If $\theta$ is computed at only one vertex, say $n_1$ ($n_2$) then put $u^{in}$ in $J_1$ ($J_2$).
 \item If all the functions in $\sucfunction \theta$ are computed only at $n_1$ ($n_2$) then put $u^{out}$ in $J_1$ ($J_2$). 
 \item If some of the functions of $\sucfunction \theta$ are computed at $n_1$ and some are computed at $n_2$ then put $u^{out}$ in $J_3.$
 
\end{enumerate}

It is easy to observe that this cut is a valid $\mincut(j_1,j_2).$ Now we compute the weight of the cut by computing $\delta(J_1), \delta(J_2).$ First note that none of the $\infty$ weight edges of $j_1,j_2$ are in the cut as corresponding sources and the sink are mapped to $J_1,J_2.$ Similarly, any vertex $u^{out}$ is mapped in $J_1$ or $J_2$ if all its successor functions are computed there. Thus, no $\infty$ weight is in $\delta(J_1), \delta(J_2)$ and the cut size is finite. Observe that the way $J$ is constructed from $\compgraph$ corresponding to all the outgoing edges of any vertex $u$ there is only one edge $(u^{in},u^{out}) \in E_J$ of same weight. This edge is in the cut constructed above iff any of the corresponding edges are exposed in $\embedding$ (points~5,~6). Hence the weight of the cut constructed above is same as that of $\embedding.$ 
\end{proof}

\begin{lemma}
 \label{lm:2cut-to-embedding}
 If there is a $\mincut(j_1,j_2)$ for the instance $\phi$ of weight $C$ then there is an embedding of $\compgraph$ on $\net$ of cost $\leq C.$
\end{lemma}
\begin{proof}
 Recall that a $\mincut(j_1,j_2)$ partitions the elements of $V_J$ into three sets $J_1,J_2,J_3.$ We create an embedding from the cut as follows: If any vertex $u^{in} \in V_J$  is in $J_1$ ($J_2$) then map $u$ at $n_1$ ($n_2$) under embedding $\embedding.$ If $u^{in}$ is in $J_3$ then map $u$ to both $n_1$ and $n_2.$ As the weight of the cut is finite, we know that all the sources of $\compgraph$  which are connected to $j_1$ (or $j_2$) the corresponding ``in" vertices are in $J_1$ (or $J_2)$. This ensures that all the sources are mapped either to $n_1$ or $n_2$ under $\embedding.$ Similarly, the sink of $\compgraph$ is in $J_2$ and thus mapped to $n_2$ under $\embedding.$ Observe that all the edges which are in $\delta(J_1)$ and $\delta(J_2)$ are exposed in the embedding $\embedding.$ Hence the cost of this embedding is same as that of the cut $C.$ As the vertices in $J_3$ are mapped at both $n_1$ and $n_2,$ there will be some redundant computations in $\embedding.$ For example some vertex $u$ might be computed at both nodes but all its successors are computed only at $n_1,$ thus making the computation at $n_2$ redundant. To get a valid embedding we need to remove such computations and removing (or pruning) such computations will only reduce the cost from $C.$ As there are only two nodes in the network checking for redundant computations for each vertex of $\compgraph$ can be done in polynomial time and thus gives an embedding $\embedding$ of cost $\leq C.$
\end{proof}

\textbf{Proof of forward direction (Theorem~\ref{thm:cut-to-cost}):} We need to prove that the minimum cost embedding has optimal embedding of cost $C$ if the \mincut\ has optimal cut of weight $C.$ Let $\embedding$ be an embedding obtained by applying the procedure on the optimal \mincut\ presented in the proof of Lemma~\ref{lm:2cut-to-embedding} with cost $C' \leq C.$ Let $C' <C.$ Then by Lemma~\ref{lm:embedding-to-2cut} we can obtain a \mincut\ of $\phi$  of weight $C'.$ But this is a contradiction to the fact that $\phi$ has the optimal cut of weight $C.$ Thus the embedding $\embedding$ obtained from the optimal cut of $\phi$ has cost $C' = C.$ 

\textbf{Proof of backward direction (Theorem~\ref{thm:cut-to-cost}):} Now we need to prove that if there is an optimal embedding of cost $C$ then $\phi$ has the optimal cut of weight $C.$ By Lemma~\ref{lm:embedding-to-2cut} we can obtain a \mincut\ for $\phi$ of weight $C$ from the optimal embedding of $\psi.$ This cut has to be the optimal cut else we can get an embedding of lesser cost than $C$ by Lemma~\ref{lm:2cut-to-embedding}.
\end{proof}

\section{Approximate Algorithms}
\label{sec:app_algo}

In Section~\ref{sec:hardness} we proved that finding a rate maximizing schedule is MAX SNP-hard. In this section we define a restricted class of embeddings and present some approximation algorithms for the corresponding maximum rate schedule problem.

\begin{definition} [\textbf{\rembedding}]
\label{def:restricted-embedding}
 A restricted embedding (\rembedding) of  $\compgraph$ on $\net$ is a function $\mathcal{E}': \compedges \mapsto \setofpaths$ which follows the following set of rules.
 \begin{enumerate}
   \item For some $\gamma \in \compedges$ if $\tail {\gamma} = \omega_i, i \in [1,\kappa]$ then 
    $\pathstart {\mathcal{E}'(\gamma)} =s_i.$ 
  \item If for some $\gamma \in \compedges$, $\head{\gamma} = \omega_p$ then  
    $ \pathend {\mathcal{E}'(\gamma)} = t.$
  \item If $\gamma_i \in \sucedges{\gamma_j}$ for some $\gamma_i,\gamma_j \in \compedges$ then 
   $\pathend{\mathcal{E}'(\gamma_j)} = \pathstart{\mathcal{E}'(\gamma_i)}.$
 \end{enumerate}
\end{definition}

Note that any intermediate function is computed only once in the network under \rembedding. 
\rembedding s are a special case of the embedding (defined in Definition~\ref{def:embedding}) and let $\setofembeddings'$ be the set of all the \rembedding s of $\compgraph$ on $\net.$ 

We can write a packing linear program, similar to CALP (presented in Section~\ref{sec:CALP}), in which the embeddings are coming from the set $\setofembeddings'$ instead of the general set of embeddings $\setofembeddings.$ Let us call this LP as R-CALP. We observe that the separation oracle of the dual of R-CALP also reduces to the problem of finding \emph{minimum cost \rembedding} problem where the cost of the \rembedding\ is defined by Equation~\eqref{eq:orig_cost}. Hence forth we refer the problem of finding the minimum cost \rembedding\ 
 by \mincostnew. It is easy to verify that Theorem~\ref{thm:approx_separation} also holds in this case giving us the following corollary.

\begin{corollary}
 \label{cor:approx_separation_rcalp}
 There is a polynomial time $\alpha$-approximation algorithm to solve R-CALP if and only if there is a polynomial time $\alpha$-approximation algorithm for solving \mincostnew\ of $\compgraph$ on $\net.$
\end{corollary}

In Section~\ref{sec:step2} we proved that minimum cost embedding problem is MAX SNP-hard  by reducing it from \maxcut\ problem. Recall that the instance of minimum cost embedding problem which we created has the optimal embedding in which one vertex of $\compgraph$ is mapped to only one vertex of $\net.$ Thus the reduction presented in Theorem~\ref{thm:flow-to-cost} also proves that solving the minimum cost \rembedding\ problem is MAX SNP-hard. In this section we present some approximation algorithms to solve \mincostnew\ problem thus giving approximate solutions for R-CALP.

We first present a version of minimum cost embedding problem which has been studied in literature and relate it to the one presented in Section~\ref{sec:step1} by Theorem~\ref{thm:old_new_relation}. Using the result of Theorem~\ref{thm:old_new_relation}  and the procedure described in the proof of Theorem~\ref{thm:approx_separation} we give a couple of algorithms to find approximate solutions of R-CALP for special classes of computation graph. 

\subsection{A version of minimum cost embedding}
\label{sec:oldcost}

A version of \mincostnew\ has been studied in literature under various names like function computation \cite{Vyavahare14,Shah13}, optimal operator placement \cite{Ying08,Bonfils04,Abrams05,Phatak10} and module placement \cite{Bokhari81,Stone77,MaryLo88,Fernandez-Baca89}. 

The cost model of this literature differs from our cost model (\mincostnew) in the following two ways ---(1)~in their cost model two outgoing edges of a vertex $\omega$ of $\compgraph$ can have different weights and, (2)~if an edge $e \in \netedges$ is used by multiple, say $z,$ outgoing edges of a vertex $\omega$ of $\compgraph$ in an embedding then while computing the cost of the embedding the weight $x(e)$ is considered $z$ times. In our cost model even if an edge $e$ is used by multiple outgoing edges of a vertex of $\compgraph$, the weight $x(e)$ is taken only once. We define their cost model more formally below.

Let $\xi_{\embedding'}^{\gamma}(e) := \mathbbm{1}\{e \in \embedding'(\gamma) \}$ be an indicator function which takes value $1$ if an edge $e$ in $\net$ is used by an edge $\gamma$ of $\compgraph$ under \rembedding\ $\embedding'.$ Then given a vector $\{x(e) | e \in \netedges\}$ and weight function $\{w(\gamma) | \gamma \in \compedges\}$ \footnote{Note that the weights in this case are defined on the edges of $\compgraph$ and outgoing edges of a vertex in $\compgraph$ can have different weights.} the cost of an \rembedding\ is defined as:
\begin{equation}
 \complement(\embedding') := \sum_{e \in \netedges} \xi_{\embedding'}(e) x(e) = \sum_{e \in \netedges} \left(\sum_{\gamma \in \compedges} \xi_{\embedding'}^{\gamma}(e) w(\gamma) \right) x(e). \label{eq:oldcost}
\end{equation}

\begin{definition}[\mincostold]
Given a network graph $\net$ with weight function $x$ on its edges, a computation graph $\compgraph$ with weight function $w$ on its edges find an \rembedding\ $\mathsf{opt}(\complement)$ such that:
 \begin{displaymath}
  \mathsf{opt}(\complement,\compgraph,\net) := \argmin_{\embedding' \in \setofembeddings'} \complement(\embedding')
 \end{displaymath}
\end{definition}

We omit $\compgraph,\net$ from the above expression when it is clear from the context and use $\mathsf{opt}(\complement)$ to represent the optimal embedding for \mincostold. Observe that $\mathsf{opt}(\complement)$ has the following properties: (1)~A vertex of $\compgraph$ is mapped to only one vertex of $\net.$ This property is imposed because of the definition of \rembedding. (2)~Every edge $\gamma$ of $\compgraph$ is mapped to the shortest path between its mapped end points in $\net$ due to the nature of the cost defined in Equation~\eqref{eq:oldcost}. 

Example~\ref{ex:newcost} below illustrates the difference between the two cost models and shows how our cost model is more natural when $\compgraph$ is a DAG.

\begin{example}
 \label{ex:newcost}
 We revisit Example~\ref{ex:intro} here. Recall that for the computation graph of Fig.~\ref{fig:ex_rate}b, $w(\gamma) =1 \forall \gamma \in \compedges.$ Let $x(e) =1 \forall e \in \netedges$ for the network shown in Fig.~\ref{fig:ex_rate}a. Then the cost of the embedding $\embedding_1$ (shown in Fig.~\ref{fig:ex_rate}c) according to Equation~\eqref{eq:orig_cost} is $C(\embedding_1) = 6$ while the cost according to Equation~\eqref{eq:oldcost} is $\complement(\embedding_1) = 7.$ This difference is due to the fact that the cost incurred over link $xz$ for the transmission of function $\theta_5$ in $\embedding_1$ is taken only once in account by Equation~\eqref{eq:orig_cost} while Equation~\eqref{eq:oldcost} considers it twice \footnote{Because of the two outgoing edges of node $\omega_5$ in $\compgraph$}. In practice the function $\theta_5$ is transmitted only once over $xz$ in $\embedding_1$ and rate computation in Example~\ref{ex:intro} does consider this.
\end{example}

Polynomial time algorithms to solve \mincostold\ problem when $\compgraph$ is a tree are available in various literature, e.g., \cite{Bokhari81,Ying08,Shah13}. Authors in \cite{Fernandez-Baca89} gave polynomial time algorithm when $\compgraph$ is $k$-tree while \cite{Vyavahare14} proves that the \mincostold\ is MAX SNP-hard for general $\compgraph.$ A polynomial time algorithm for a layered $\compgraph$ is presented in \cite{Vyavahare14}. \mincostold\ problem is also related to two well studied problems like Multiterminal cut and $0$-extension problem. We explain the relation with these problems below.

\paragraph{Connection to Multiterminal cut problem} \mincostold\ problem, when $\net$ is a complete graph of $k$ terminals with weights $x(e) =1 \forall e \in \netedges,$ is equivalent to a well known NP-complete problem \textit{Multiterminal Cut} \cite{Dahlhaus94}. The Multiterminal Cut problem is defined as follows: Given a graph $\compgraph=(\compnodes,\compedges)$ with weights $w(\gamma)$ on its edges and a set of $k$ of its vertices, divide the graph $\compgraph$ into $k$ parts such that there is only one terminal in each part and the sum of the weights of the edges across these parts is minimum. In other words, Multiterminal Cut problem asks for a \rembedding\ $\embedding$ of $\compgraph$ on a complete graph $\net = (\netnodes,\netedges)$  with $|\netnodes| = k$ and $x(e) =1 \forall e \in \netedges$ such that cost $\complement(\embedding)$ is minimum. Refer to \cite{Vyavahare14} for the details of this reduction which proves that \mincostold\ problem is MAX SNP-hard even if the number of terminals $k$ and the weights on the edges $w(\gamma)$ are constant.

\paragraph{Connection to $0$-extension problem} When the network graph $\net$ is a complete graph with $k$ vertices but with arbitrary edge weights then the problem \textit{$0$-extension} can be seen as a special case of \mincostold\ problem. $0$-extension problem was first introduced by \cite{Karzanov98} and is defined as follows: Given a graph $\compgraph=(\compnodes,\compedges)$ with non negative edge weights $w(\gamma)$ on its edges and a metric $d$ defined on a subset $T \subseteq \compnodes,$  find an assignment $\embedding$ of every $\omega \in \compnodes$ on $\embedding(\omega) \in T$ such that $\embedding(\omega) = \omega \forall \omega \in T$ and the cost $\sum_{(\omega_1,\omega_2) \in \compedges} w(\omega_1,\omega_2) d(\embedding(\omega_1),\embedding(\omega_2))$ is minimum. In other words, $0$-extension problem asks for a \rembedding\ $\embedding$ of $\compgraph$ on a complete graph $\net = (\netnodes,\netedges)$  with $|\netnodes| = |T|$ and $\{x(e) | e \in \netedges\}$ where $x(e)$ imposes a metric on $\netnodes$ such that the cost $\complement(\embedding)$ is minimum. The $0$-extension problem is a well studied problem and we refer the readers to \cite{Karloff06} for a detailed review of the results available in the literature. Authors in \cite{Karloff06} proved that for every $\epsilon >0,$ there is no polynomial time $O((\log p)^{1/4 -\epsilon})$- approximate algorithm for 0-extension unless NP $\subseteq$ DTIME$(p^{poly(\log p)})$ where $p$ is the number of vertices in $\compgraph$ with the maximum degree of any vertex and the weight of an edges as $poly(\log p).$ This result also holds for \mincostold\ problem as $0$-extension is a special case of it.

Next we prove a relation between the \mincostold\ and \mincostnew\ problems.
\begin{theorem}
 \label{thm:old_new_relation}
 Given a network graph $\net$ with weight function $x$ on its edges and a computation graph $\compgraph$ with weight function $w$ on its edges the optimal solution of \mincostold\ problem gives a $D$-approximation of \mincostnew\ problem where $D$ is the maximum out-degree of any vertex in $\compgraph.$
\end{theorem}
\begin{proof}
Recall that the cost of a \rembedding\ of $\compgraph$ on $\net$ is computed by Equations~\eqref{eq:orig_cost}, ~\eqref{eq:oldcost} in \mincostold\ (denoted by $\complement(\embedding)$) and \mincostnew\ (denoted by $C(\embedding)$) problem, respectively. Let us consider a computation graph $\compgraph$ in which outgoing edges of any vertex are not more that $D.$ As seen earlier weight of an edge $e$ in $\net$ considered multiple times if it is used by multiple outgoing edges of a vertex of $\compgraph$ in an embedding $\embedding$ while computing $\complement(\embedding)$ but it is considered only once for computation of $C(\embedding).$ Thus, for any embedding $\embedding,$ $C(\embedding) \leq \complement(\embedding).$ By the same argument if the maximum number of outgoing edges of any vertex of $\compgraph$ is $D$ then an edge $e$ of $\net$ can be used at most $D$ times by outgoing edges of any vertex. Thus the cost coming from mapping of outgoing edges of a vertex of $\compgraph$ on any edge $e$ of $\net$ in $\complement(\embedding)$ could be at most $D$ times the cost coming from $e$ in $C(\embedding)$ which implies that $\complement(\embedding) \leq D C(\embedding).$ Combining both the arguments we have,
\begin{equation}
 \complement(\embedding) \leq D C(\embedding) \leq D \complement(\embedding). \label{eq:oldnewcost}
\end{equation}

Let $\embedding_1$ and $\embedding_2$ be the optimal solutions of \mincostold\ and \mincostnew\ problem respectively. Then, $C(\embedding_2) \leq C(\embedding_1) \leq \complement(\embedding_1) \leq \complement(\embedding_2) \leq D C(\embedding_2),$ where first and fourth inequalities are due to the definitions of $\embedding_1,\embedding_2$ and  second and third inequalities are due to Equation~\ref{eq:oldnewcost}. Thus,

\begin{displaymath}
 C(\embedding_2) \leq C(\embedding_1) \leq D C(\embedding_2).
\end{displaymath}
This proves the theorem.
\end{proof}

This implies that an algorithm which gives an $\alpha$-approximate solution for \mincostold\ problem also gives an $\alpha D$-approximate solution for \mincostnew\ problem. Recall that by Theorem~\ref{thm:approx_separation} there is an $\alpha$-approximation algorithm for solving R-CALP if and only if there is an $\alpha$-approximation algorithm for \mincostnew\ problem. Combining this fact with the hardness result for $0$-extension in \cite{Karloff06} we get the following result.

\begin{corollary}
 \label{cor:oldcost_hardness}
 Given an arbitrary network graph $\net$ and a computation graph $\compgraph$ with $p$ vertices and the maximum degree of a vertex and the maximum weight on an edge in $\compgraph$ is $poly(\log p),$ for any $\epsilon >0,$ there is no polynomial time approximation algorithm with approximation ratio of $O(poly(\log p)(\log p)^{1/4 -\epsilon})$ for solving R-CALP unless NP $\subseteq$ DTIME$(p^{poly(\log p)}).$
\end{corollary} 

Now we present polynomial time approximate algorithms for special classes of computation graph $\compgraph.$

\subsection{When $\compgraph$ is a layered graph}
\label{sec:layer}

In this section we consider the case when $\compgraph$ is a layered graph. An example of layered graph is shown in Fig.~\ref{fig:layer}. We assume that there are $r$ layers and each layer has at most $W$ vertices. We number layers from $\{1,\ldots,r\}$ and vertices of a layer $l$ by $\{\omega_{1l},\ldots,\omega_{Wl}\}.$ An edge $\{\omega_{ai},\omega_{bj}\}$ is present only if $j = i+1.$ We also assume that the sink vertex is present on the $r$-th layer. Note that this implies that the out-degree of any vertex in a layered graph is at most $W.$ Commonly used layered computation graphs are butterfly structure of fast Fourier transform (FFT), correlation function and functions of Boolean data in Sum of Product (or Product of Sum) form.

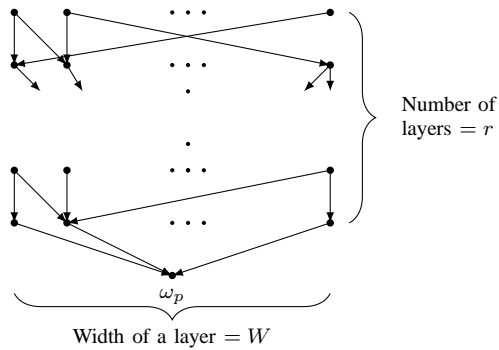
\begin{figure}[t]
     \centering
\begin{tikzpicture}[>=latex,scale=0.7]
\footnotesize
\foreach \x in {-3,-2,3} {
  \foreach \y in {-1,0,2,3} {
    \fill (\x,\y) circle (0.07cm);
  }  
}    
\foreach \x in {0,0.3,0.6} {
  \foreach \y in {-1,0,2,3} {
    \fill (\x,\y) circle (0.04cm);
  }  
} 
\foreach \y in {0,0.5,1.5}
  \fill (0.3,\y) circle (0.04cm);
\fill (0,-2) circle (0.07cm);

\draw[->] (-3,3) -- (-3,2);
\draw[->] (-3,3) -- (-2,2);  
\draw[->] (-2,3) -- (-2,2);     
\draw[->] (3,3) -- (-3,2);
\draw[->] (-2,3) -- (3,2);
\draw[->] (-3,2) -- (-2.5,1.5);
\draw[->] (-2,2) -- (-1.7,1.5);
\draw[->] (3,2) -- (2.5,1.5);
\draw[->] (3,2) -- (3,1.5);
\draw[->] (-3,0) -- (-3,-1);
\draw[->] (-3,0) -- (-2,-1);  
\draw[->] (-2,0) -- (-2,-1);     
\draw[->] (3,0) -- (-2,-1);
\draw[->] (3,0) -- (3,-1);
\draw[->] (-3,-1) -- (0,-2);
\draw[->] (-2,-1) -- (0,-2);
\draw[->] (3,-1) -- (0,-2);
\node at (0,-2.4){$\omega_p$};    
\draw [decorate,decoration={brace,amplitude=10pt},xshift=-4pt,yshift=0pt]
(3.5,3) -- (3.5,-1) node [black,midway,xshift=0.6cm,text width=2cm,right] {Number of layers $= r$};
\draw [decorate,decoration={brace,amplitude=10pt},xshift=0pt,yshift=-4pt]
(3,-2.2) -- (-3,-2.2) node [black,midway,yshift=-0.6cm] {Width of a layer $= W$};

\end{tikzpicture}

    \caption{A layered computation graph}
    \label{fig:layer}
\end{figure}
  
A polynomial time algorithm is presented in \cite{Vyavahare14} which solves \mincostold\ problem for a layered $\compgraph$ and an arbitrary $\net.$ This algorithm takes $O(rn^{2W})$ time where $n$ is the number of vertices in $\net.$ Theorem~\ref{thm:old_new_relation} implies that this algorithm is a $2W$-approximation algorithm for \mincostnew\ problem. Recall that \mincostnew\ problem is the separation oracle for the dual of R-CALP and by the method described in Section~\ref{sec:step2} we can solve the R-CALP by using \mincostnew\ solution. This leads us to the following result.

\begin{corollary}
 \label{cor:boundedwidth_layer}
 Given an arbitrary network graph $\net$ with non-negative capacities on its edges and a layered computation graph $\compgraph$ with $r$ layers and at most $W$ vertices at each layer, there is a polynomial time $W$-approximation algorithm to solve R-CALP.
\end{corollary}

The complexity of the algorithm of Corollary~\ref{cor:boundedwidth_layer} is exponential in the width of any layer thus the algorithm cannot be applied to layered graphs with unbounded width. We now present a procedure to get an $O(F)$-approximation of \mincostold\ problem for a computation graph $\compgraph$ which has a spanning tree $\mathcal{T}$ such that any edge of $\mathcal{T}$ is a part of at most $O(F)$ fundamental cycles. A fundamental cycle is a cycle created by adding an edge from $\compgraph$ to $\mathcal{T}.$ For every edge $uv \notin \mathcal{T}$ there is a unique such cycle created by the edges of $\mathcal{T}$ and $uv.$
\begin{theorem}
 \label{thm:unboundedwidth_layer}
 Given an arbitrary network $\net$ and a computation graph $\compgraph$ with a spanning tree $\mathcal{T}$ such that any edge of $\mathcal{T}$ is a part of at most $O(F)$ fundamental cycles, there is a polynomial time $O(F)$-approximation algorithm to solve \mincostold\ problem.
\end{theorem}
\begin{proof}
Let $\mathcal{T}$ be the spanning tree of $\compgraph$ such that any of its edge is a part of at most $O(F)$ fundamental cycles. Recall that polynomial time algorithms to find optimal solution for \mincostold\ when the computation graph is a tree are known in the literature \cite{Bokhari81,Ying08}. Using any of the algorithms available in \cite{Bokhari81,Ying08} we can find the optimal solution of \mincostold\ for $\mathcal{T}$ on $\net.$ Let this optimal \rembedding\ for $\mathcal{T}$ be $\mathsf{opt}(\mathcal{T})$ with cost $\complement(\mathcal{T}).$ Note that the \rembedding\ $\mathsf{opt}(\mathcal{T})$ gives a mapping for each vertex of $\compgraph$ on $\net.$ We create an \rembedding\ $\mathcal{X}$ for $\compgraph$ from $\mathsf{opt}(\mathcal{T})$ as follows: Map an edge $(u,v) \in \compgraph$ to the shortest path between its mapped end points in $\mathsf{opt}(\mathcal{T}).$ In this way the edges of $\compgraph$ which are in $\mathcal{T}$ are mapped to the same paths as in $\mathsf{opt}(\mathcal{T}).$ It is easy to observe that it is a valid \rembedding\ for $\compgraph$ with cost $\complement(\mathcal{X}).$ Let the optimal solution of \mincostold\ problem for $\compgraph$ on $\net$ be $\mathsf{opt}(\compgraph)$ with cost $\complement(\mathsf{opt}(\compgraph)).$ 
It is easy to observe that the mapping of the edges of $\compgraph$ which are in $\mathcal{T}$ under the \rembedding\ $\mathsf{opt}(\compgraph)$ gives a valid \rembedding\ of $\mathcal{T}$ on $\net.$ Thus, $\complement(\mathcal{T}) \leq \sum_{uv \in \mathcal{T}} \complement_{uv}(\mathsf{opt}(\compgraph))  \leq \sum_{uv \in \mathcal{T}} \complement_{uv}(\mathsf{opt}(\compgraph)) + \sum_{uv \notin \mathcal{T}} \complement_{uv}(\mathsf{opt}(\compgraph)) \leq \complement(\mathsf{opt}(\compgraph)).$ 
Also, by the definition of $\mathsf{opt}(\compgraph)$ and $\mathcal{X}$ we get $\complement(\mathsf{opt}(\compgraph)) \leq \complement(\mathcal{X}).$ 

The cost of $\mathcal{X}$ can be written as $\complement(\mathcal{X}) = \sum_{uv \in \mathcal{T}} \complement_{uv}(\mathcal{X}) + \sum_{uv \notin \mathcal{T}} \complement_{uv}(\mathcal{X}) = \complement(\mathcal{T}) + \sum_{uv \notin \mathcal{T}} \complement_{uv}(\mathcal{X}).$
Note that for each $uv \notin \mathcal{T}$ there is a path $\sigma_{uv} \in \mathcal{T}.$ As an edge $uv \notin \mathcal{T}$ is mapped to the shortest distance between its mapped end points in $\mathcal{X}$ we get, 
\begin{displaymath}
 \sum_{uv \notin \mathcal{T}} \complement_{uv}(\mathcal{X}) \leq \sum_{uv \notin \mathcal{T}} \left(\sum_{e \in \sigma_{uv}} \complement_e(\mathcal{T}\right) \leq O(F) \complement(\mathcal{T}),
\end{displaymath}
where the last inequality is due to the property of $\mathcal{T}.$ Finally we get, $\complement(\mathcal{X}) \leq \complement(\mathcal{T}) + O(F) \complement(\mathcal{T}) \leq O(F) \complement(\mathcal{T}) \leq O(F) \complement(\mathsf{opt}(\compgraph)).$

This proves that the \rembedding\ $\mathcal{X}$ is an $O(F)$-approximation of $\mathsf{opt}(\compgraph).$
\end{proof}

Using this algorithm with the procedure described in Theorem~\ref{thm:approx_separation} we get the following result.
 
\begin{corollary}
 \label{cor:unboundedwidth_layer}
 Given an arbitrary network graph $\net$ with non-negative capacities on its edges and a computation graph $\compgraph$ with a spanning tree whose any edge is a part of at most $O(F)$ fundamental cycles, there is a $O(FD)$-approximation algorithm to solve R-CALP where $D$ is the maximum out-degree of any vertex in $\compgraph.$
\end{corollary}

An example of such a graph is the computation graph for fast Fourier transform (FFT). A FFT graph for $\kappa$ input sources can be represented by a layered graph of $r= \log (\kappa)$ layers with $W=\kappa$ vertices on each layer. Fig.~\ref{fig:fft}a shows an FFT computation graph for $4$ sources and its spanning tree is shown in Fig.~\ref{fig:fft}b. It is easy to observe that in such a spanning tree of any FFT structure any edge is a part of at most $O(\log (\kappa))$ fundamental cycles. This gives a $O(\log (\kappa))$-approximation for R-CALP with $k$-point FFT computation graph.

\begin{figure}[t]
    \centering
    \subfloat[]{
\begin{tikzpicture}[>=latex,scale=0.7]

\footnotesize
\foreach \x in {-3,-1,1,3}
  \foreach \y in {3,2,1}
    \fill (\x,\y) circle (0.07cm);
    
 \fill (0,0) circle (0.07cm);
 \fill (0,-1) circle (0.07cm);
 \node at (0,-1.3) {$\omega_p$};
 \node at (-3,3.2) {$s_1$};
 \node at (-1,3.2) {$s_2$};
 \node at (1,3.2) {$s_3$};
 \node at (3,3.2) {$s_4$};

 \draw [->] (0,0) -- (0,-1);
 \foreach \x in {-3,-1,1,3} {
  \draw[->] (\x,1) -- (0,0);
   \foreach \y in {3,2} {
    \draw[->] (\x,\y) -- (\x,\y-1);
    }
  }   
  \draw[->] (-3,3) --(-1,2);
  \draw[->] (-1,3) -- (-3,2);
  \draw[->] (3,3) --(1,2);
  \draw[->] (1,3)--(3,2);
  \draw[->] (-3,2) --(1,1);
  \draw[->] (-1,2)--(3,1);
  \draw[->] (1,2)--(-3,1);
  \draw[->] (3,2) --(-1,1);
          
\end{tikzpicture}
}
\hspace*{10pt}
\subfloat[]{
\begin{tikzpicture}[>=latex,scale=0.7]

\footnotesize
\foreach \x in {-3,-1,1,3}
  \foreach \y in {3,2,1}
    \fill (\x,\y) circle (0.07cm);
    
 \fill (0,0) circle (0.07cm);
 \fill (0,-1) circle (0.07cm);
 \node at (0,-1.3) {$\omega_p$};
 \node at (-3,3.2) {$s_1$};
 \node at (-1,3.2) {$s_2$};
 \node at (1,3.2) {$s_3$};
 \node at (3,3.2) {$s_4$};
 
 \draw [->] (0,0) -- (0,-1);
 \foreach \x in {-3,-1,1,3} {
  \draw[->] (\x,1) -- (0,0);
   \foreach \y in {3,2} {
    \draw[->] (\x,\y) -- (\x,\y-1);
    }
  }   
\end{tikzpicture}
}

    \caption{(a) FFT structure for $4$ sources. (b) A spanning tree of graph shown in (a)}
    \label{fig:fft}
\end{figure}
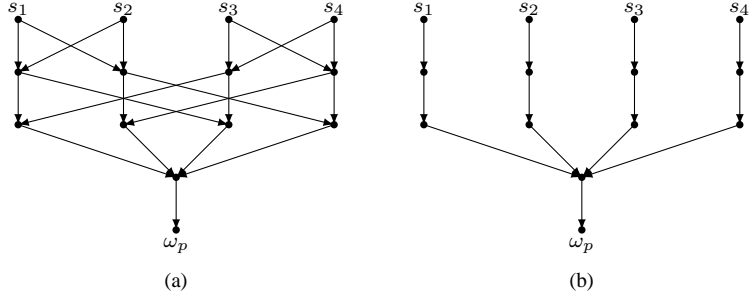

\subsection{QIP for \mincostold\ and its LP relaxation}
\label{sec:lp}

In this section we present a quadratic integer program to solve \mincostold\ problem and its linear programming relaxation. A similar quadratic integer program for \mincostold\ has been presented in \cite{Vyavahare14a}. Then we show how the algorithms of \cite{Calinescu01} for $0$-extension can be extended to get approximate algorithms for \mincostold\ which in turn gives an approximate algorithm for R-CALP. 

The quadratic integer program for \mincostold\ problem is shown below. It is easy to verify that the objective function is same as Equation~\eqref{eq:oldcost} where $d(u,v)$ is the shortest distance between vertices $u,v$ in the network graph. Recall that in an \rembedding\ a vertex of the computation graph is mapped to only one vertex in the network graph. Thus for each vertex $\alpha \in \compnodes, u \in \netnodes$ we define a binary variable $x_{\alpha u},$ which takes the value one if and only if $\alpha$ is mapped to $u$ in the embedding which minimizes the objective function. The embedding constraints ensure that each vertex $\alpha$ is mapped only to one of the vertices in $\netnodes.$ Likewise the source and sink constraints ensure that the sources and sink of computation graph are mapped to the corresponding sources and sink in the network graph.

\centerline{\rule{\columnwidth}{0.75pt}} \textbf{Quadratic Integer Program for \mincostold} \cite{Vyavahare14a}

\textbf{Objective:}$
  \min \sum\limits_{(\alpha,\beta) \in \compedges}  w(\alpha,\beta) \left(\sum\limits_{u,v \in \netnodes}  x_{\alpha u} d(u,v) x_{\beta v}\right) $ 
  \textbf{subject to}
\begin{enumerate}
\item Source constraints
 \begin{displaymath}
  x_{\alpha u} = 1 \mbox{ if } \alpha = \omega_i \mbox{ and } u= s_i
  \forall i \in [1,\kappa]
 \end{displaymath}
 \item Sink Constraint
 \begin{displaymath}
  x_{\alpha u} = 1 \mbox{ if } \alpha = \omega_p \mbox{ and } u= t
 \end{displaymath}
 \item \rembedding\ constraints
 \begin{displaymath}
   \sum\limits_{u \in \netnodes} x_{\alpha u} = 1 \mbox{ } \forall
   \alpha \in \compnodes
 \end{displaymath}
 \item Binary constraints
   \begin{displaymath}
     x_{\alpha u} \in \{0,1\} \mbox{ } \forall \alpha \in \compnodes,
     u \in \netnodes
 \end{displaymath}
\end{enumerate}
\centerline{\rule{\columnwidth}{0.75pt}}

Note that the objective function of the above QIP is a quadratic function of the binary variables $x_{\alpha u}.$ We relax this QIP into a linear program by using the concept of \textit{earthmover distance metric} which is very similar to the relaxation presented for $0$-extension problem in \cite{Archer04}. Recall that the shortest distance $d(u,v)$ forces a metric on the vertex set $\netnodes$ of the network graph and $|\netnodes| = n.$ Given a metric $(\netnodes,d)$ on a set $V$ the earthmover distance extends the metric to the probability distributions over $\netnodes.$ If any probability distribution $\bar{a} := \{a_1,\ldots,a_n\}$ over $\netnodes$ is seen as $a_i$ amount of dirt piled on $i \in \netnodes$ then the earthmover distance between $\bar{a}$ and a distribution $\bar{b}:= \{b_1,\ldots,b_n\}$ is the minimum cost of moving the dirt from configuration $\bar{a}$ to $\bar{b}.$ The earthmover distance, $d_{EM}(a,b),$ between two distributions can be found by the following flow problem.

\centerline{\rule{\columnwidth}{0.75pt}}
\textbf{Objective:} $d_{EM}(a,b)= \min \sum\limits_{u,v \in \netnodes} d(u,v) f_{uv}$ \textbf{subject to:}
\begin{enumerate}
 \item $\sum\limits_{v \in \netnodes} f_{uv} = a_u \mbox{ } \forall u \in \netnodes$
 \item $\sum\limits_{u \in \netnodes} f_{uv} = b_v \mbox{ } \forall v \in \netnodes$
 \item $f_{uv} \geq 0 \mbox{ } \forall u,v \in \netnodes$
\end{enumerate}
\centerline{\rule{\columnwidth}{0.75pt}}

In the flow problem above the variable $f_{uv}$ represents the amount of dirt to be moved from $u$ to $v$ while going from configuration $\bar{a}$ to $\bar{b}.$ 

To get the LP relaxation for the QIP we first replace the binary constraints by $0 \geq x_{\alpha u} \geq 1$ for each $\alpha \in \compnodes, u \in \netnodes$ except for the sources and sink. Then we replace the term $x_{\alpha u} x_{\beta v}$ in the objective function by a variable $y_{\alpha u \beta v}$ resulting in the following objective function.
\begin{displaymath}
\min \sum\limits_{(\alpha,\beta) \in \compedges}  w(\alpha,\beta) \left(\sum\limits_{u,v \in \netnodes}  y_{\alpha u \beta v} d(u,v) \right) 
\end{displaymath}

Multiplying the \rembedding\ constraint by $x_{\beta v}$ and $x_{\alpha u}$ appropriately on both sides we get the new constraints for the variables $y_{\alpha u \beta v}$ as---(1) $\sum\limits_{u \in \netnodes} y_{\alpha u \beta v} = x_{\beta v} \mbox{ } \forall
   \alpha \in \compnodes, v \in \netnodes \mbox{ and } (\alpha,\beta) \in \compedges,$ (2)~ $\sum\limits_{v \in \netnodes} y_{\alpha u \beta v} = x_{\alpha u} \mbox{ } \forall
   \beta \in \compnodes, u \in \netnodes \mbox{ and } (\alpha,\beta) \in \compedges.$
   
Let $x_{\alpha} := \{x_{\alpha 1},\ldots, x_{\alpha n}\}$ be an $n$-dimensional vector where an element $x_{\alpha i}$ corresponds to the variable $x_{\alpha i}$ for $i \in \netnodes.$ Along with the \rembedding\ constraints $x_{\alpha}$ for each $\alpha \in \compnodes$ can be seen as a probability distribution over the set of network vertices $\netnodes$ and the variable $y_{\alpha u \beta v}$ can be seen as the flow variables corresponding to flow problem to solve the earthmover distance between the configuration $x_{\alpha}$ and $x_{\beta}$ for each $(\alpha,\beta) \in \compedges.$ Thus, $\min \sum\limits_{u,v \in \netnodes}  y_{\alpha u \beta v} d(u,v) = d_{EM}(x_{\alpha},x_{\beta})$ and we can write the LP relaxation as follows:

\centerline{\rule{\columnwidth}{0.75pt}} \textbf{Earthmover based linear program for \mincostold}

\textbf{Objective:}$
  \min \sum\limits_{(\alpha,\beta) \in \compedges}  w(\alpha,\beta) d_{EM}(x_{\alpha},x_{\beta}) $ 
  \textbf{subject to}
\begin{enumerate}
\item Source constraints
 \begin{displaymath}
  x_{\alpha u} = 1 \mbox{ if } \alpha = \omega_i \mbox{ and } u= s_i
  \forall i \in [1,\kappa]
 \end{displaymath}
 \item Sink Constraint
 \begin{displaymath}
  x_{\alpha u} = 1 \mbox{ if } \alpha = \omega_p \mbox{ and } u= t
 \end{displaymath}
 \item \rembedding\ constraints
 \begin{align*}
   \sum\limits_{u \in \netnodes} x_{\alpha u} &= 1 \mbox{ } \forall
   \alpha \in \compnodes 
 \end{align*}
 \item Non negativity constraints
   \begin{align*}
     0 \leq x_{\alpha u} &\leq 1  \mbox{ } \forall \alpha \in \compnodes,
     u \in \netnodes 
   \end{align*}
\end{enumerate}
\centerline{\rule{\columnwidth}{0.75pt}}

Note that we are not writing the flow constraints $y_{\alpha u \beta v}$ corresponding to $x_{\alpha},x_{\beta}$ here but they are considered in computing $d_{EM}(x_{\alpha},x_{\beta})$ while solving this LP. 

 Let $\mathsf{opt}(LP)$ and $\mathsf{opt}(QIP)$ be the optimal objective function values of the LP relaxation and QIP for \mincostold\ respectively. Observe that any solution of the QIP for \mincostold\ is also a solution of this LP thus, $\mathsf{opt}(LP) \leq \mathsf{opt}(QIP).$ If we can find a polynomial time rounding procedure which rounds the solution corresponding to $\mathsf{opt}(LP)$ to a QIP solution $x$ such that objective function value $\mathsf{sol}(x)$ of $x$ is: $\mathsf{sol}(x) \leq \alpha \mathsf{opt}(QIP).$ Then we have an $\alpha$-approximation solution for the \mincostold\ problem. 

Authors in \cite{Calinescu01} gave two randomized rounding algorithms for $0$-extension problem where the LP relaxation is based on the \textit{semi-metric} concept. First rounding procedure of \cite{Calinescu01}
gives a $O(\log(|T|))$-approximation for an arbitrary graph $\compgraph=(\compnodes,\compedges)$ where $T \subseteq \compnodes$ on which the metric is given. Recall that the $0$-extension problem can be seen as a special case of \mincostold\ problem with the network graph $\net=(\netnodes,\netedges)$ as a complete graph on vertices of $T$ with edges following the given metric and the computation graph as $\compgraph.$ The semi-metric LP relaxation allows the mapping of vertices of $\compgraph$ on an arbitrary metric containing the given metric. The semi-metric LP relaxation cannot be directly extended to \mincostold\ problem but the rounding algorithms of \cite{Calinescu01} work for our earthmover based LP relaxation. Thus an instance of \mincostold\ problem in which number of vertices in $\net$ are equal to the number of sources and sink (in other words, there are no intermediate nodes in $\net$ and $|\netnodes| = |T|)$ the first rounding procedure of \cite{Calinescu01} will give an $O(\log(|\netnodes|))$-approximation. In general for any \mincostold\ instance $|\netnodes|> |T|.$ We applied the rounding procedure of \cite{Calinescu01} to a general instance of \mincostold\ and got an $O(\log(|\netnodes|))$-approximation for that as well. Recall that the optimal solution of earthmover LP gives a $|\netnodes| =n$ length vector $x_\alpha =\{x_{\alpha 1},\ldots,x_{\alpha n}\}$ for each vertex $\alpha \in \compnodes.$ The vector $x_{\alpha}$ is a probability distribution over $\netnodes,$ where an element $x_{\alpha u}$ represents the probability with which vertex $\alpha$ of $\compgraph$ can be mapped to $u$ of $\net.$ Thus each element of it may have fractional value except for the sources and sink vectors which have integral values due to the corresponding constraints. Let $x_u := \{0,\ldots,1,0,\ldots,0\}$ be the \emph{integral} probability distribution over $\netnodes$ in which the whole mass is concentrated on the vertex $u \in \netnodes.$ For finding an integral solution corresponding to fractional solution obtained by LP, the rounding procedure first finds a subset of $\netnodes$ which is closest to $x_{\alpha}$ by finding the earthmover distance $d_{EM}(x_{\alpha},x_u) \forall u \in \netnodes.$ Then parsing all the vertices of $\netnodes$ from a random permutation of $\netnodes$ it assigns a vertex $\alpha$ to a vertex $u$ of $\netnodes$ if it is \emph{close} \footnote{Here \emph{close} is defined by a random parameter $\delta \in [1,2)$ and $\alpha$ is assigned to $u$ if $u$ is the first vertex in the permutation which is within distance $\delta$ from the subset found earlier for $\alpha.$} to the subset found earlier for $\alpha.$ Carrying out the analysis along the lines of \cite{Calinescu01} we observe that this rounding procedure gives a solution $x$ of QIP such that $\mathsf{sol}(x) \leq O(\log(n)) \mathsf{opt}(QIP).$ Combining this with the results of Theorems~\ref{thm:old_new_relation},~\ref{thm:approx_separation} we get the following result.

\begin{corollary}
\label{cor:logn-approx}
  Given an arbitrary network graph $\net$ with non-negative capacities on its edges and a  computation graph $\compgraph$ in which the out-degree of any vertex is at most $D$ there is a polynomial time $O(D\log n )$-approximation algorithm to solve R-CALP, where $n$ is the number of vertices in $\net.$
\end{corollary}

In the second rounding procedure of \cite{Calinescu01} authors exploit the structural properties of the given graph $\compgraph$ and give an $O(1)$-approximation when $\compgraph$ is planar. A common example of a planar computation graph is of the \emph{correlation function.} A correlation function over $\kappa$ sources is defined as: $f=\sum_{i=1}^{\kappa-1}x_{i} x_{i+1}.$ Observe that it can be represented as a planar layered graph. The second rounding procedure of \cite{Calinescu01} can also be applied to our earthmover LP. The analysis for this rounding procedure only depends on the structure of the graph $\compgraph$ and not on the number of vertices of $\net$ thus the same analysis also works for our case also. This leads to the following result.

\begin{corollary}
\label{cor:planar-approx}
  Given an arbitrary network graph $\net$ with non-negative capacities on its edges and a planar computation graph $\compgraph$ in which the out-degree of any vertex is at most $D$   there is a polynomial time $O(D)$-approximation algorithm to solve R-CALP.
\end{corollary}

The approximation algorithms described in this section are summarized in Table~\ref{tb:approximate_results}.
\begin{table}[ht]
\centering
 \resizebox{\linewidth}{!}{
\begin{tabular}{|c|c|c|}\hline
\textbf{Computation Graph ($\compgraph)$} & \textbf{Approximation Factor} & \textbf{Result} \\ \hline
Layered graph with constant width ($W = O(1)$) & $O(W)$ & Corollary~\ref{cor:boundedwidth_layer} \\ \hline
Graph with a spanning tree in which every edge is a part of $O(F)$ fundamental cycles & $O(FD)$ & Corollary~\ref{cor:unboundedwidth_layer} \\ \hline
Arbitrary graph with $D$ degree of any vertex & $O(D\log n)$ & Corollary~\ref{cor:logn-approx} \\ \hline
Planar graph with $D$ degree of any vertex & $O(D)$ & Corollary~\ref{cor:planar-approx} \\ \hline
\end{tabular}
}
\caption{Approximation Algorithms of R-CALP for a specific computation graph ($\compgraph)$ and arbitrary network graph ($\net)$ with $n$ vertices}
\label{tb:approximate_results}
\end{table}

\section{Discussion}
\label{sec:discussion}

In this work we studied the problem of finding maximum rate schedule to compute a function $f$ on a capacitated network $\net$ when the computation schema for $f$ is given by a DAG, $\compgraph.$ We proved that solving this problem is MAX SNP-hard in general and presented some polynomial time approximate algorithms for a restricted class of schedules. Algorithmic lower bounds have been obtained for many known NP-hard problems under the exponential running time assumption for algorithms for  satisfiability (SAT) problem \cite{Lokshtanov11}. These assumptions are called \emph{Exponential Time Hypothesis} (ETH) and \emph{Strong Exponential Time Hypothesis} (SETH). SETH and ETH have led to tight lower bounds for several graph problems on bounded treewidth graphs (with running time being exponential in treewidth). It will be interesting to investigate the maximum rate problem under ETH and SETH.  We provided some polynomial time approximate algorithms for minimum cost embedding problem here, but we did not investigate the \emph{parameterized complexity} \cite{Downey99} of the problem. Possible parameters for the minimum cost embedding problem could be the treewidth of $\compgraph,$ or the number of sources in $\compgraph.$ Finding algorithms which are exponential only in the size of the fixed parameter but polynomial in the size of input can enhance the understanding of the minimum cost embedding problem and help us design better algorithms for a general class of $\compgraph.$

\section{Acknowledgment}
\label{sec:ack}
The authors would like to thank Sundar Vishwanathan for the idea of Theorem~\ref{thm:unboundedwidth_layer}.

\bibliographystyle{plain}
\bibliography{function-computation.bib}

\appendices

\section{Properties of an embedding}
\label{app:embedding_property}

Recall that an embedding maps an edge $\gamma$ to a set of paths such that the function carried by it, say $\theta,$ is computed by start node of the path and is used by the end node of the path to generate the successor function. Thus any edge in $\compgraph$ which starts from a source vertex $\omega_i$ should be mapped to a path in $\net$ which starts from $s_i$ (item~1 of Definition~\ref{def:embedding}). Similarly, any incoming edge of sink vertex $\omega_p \in \compnodes$ should be mapped to paths which end at the sink $t \in \netnodes$ (item~2 of Definition~\ref{def:embedding}). According to a computation event in $\net$ any vertex $u \in \netnodes$ can compute a symbol of a function $\theta$ at time $\tau$ if the corresponding symbols of all its predecessor functions are available at $u.$ Thus, for every edge $\gamma$ of $\compgraph,$ the end points of one of the paths to which its predecessor edges are mapped should be the same as the start point of a path to which $\gamma$ is mapped and vice versa (item~3 of Definition~\ref{def:embedding}).

Fig.~\ref{fig:emb_type_valid} shows some valid path structures to embed an edge $\gamma \in \compedges$ in $\net.$ In the structures shown in  Figs.~\ref{fig:emb_type_valid}b and c, the function $\theta$ is computed only once (by node $a$) but used at two different nodes to compute the same successor function. Such an embedding is shown in Fig.~\ref{fig:example}d of Example~\ref{ex:embedding}. Similarly, in embedding structure of Fig.~\ref{fig:emb_type_valid}d function $\theta$ is computed at two nodes and used by two different nodes in $\net.$ 

In any valid embedding same symbol of any function $\theta$ should not be carried by an edge in $\net$ multiple times or received by a node multiple times (item~4,5 of Definition~\ref{def:embedding}). Figs.~\ref{fig:emb_type_invalid}b,c,d correspond to the structures in which the function $\theta$ is carried multiple times by an edge (edge $(c,d)$ in Figs.~\ref{fig:emb_type_invalid}b,c) or received multiple times by a node (node $c$ in Fig.~\ref{fig:emb_type_invalid}d). These structures will not occur in any valid embedding.

 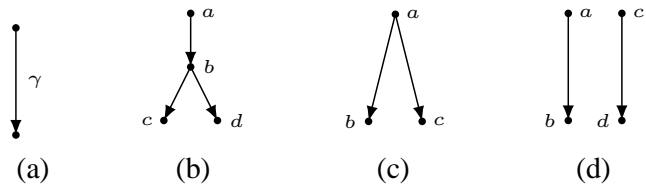
\begin{figure}[h]
 \begin{center}
    \resizebox{0.5\textwidth}{!}{
        \subfloat[]{
\begin{tikzpicture}[>=latex,scale=0.5]
\scriptsize
 \tikzstyle{every node} = [circle,draw=black]
 \tiny
 \fill (0,1) circle (0.07cm);
 \fill (0,-1) circle (0.07cm);
 \draw[->] (0,1) -- (0,-1) node [draw=none,midway,right] {$\gamma$};
 \node [draw=none] at (0.5,-0.5) {};
 \end{tikzpicture}
}
\hspace{10pt}
  \subfloat[]{
\begin{tikzpicture}[>=latex,scale=0.5]
\tiny
 \tikzstyle{every node} = [circle,draw=black]
 \fill (0,1) circle (0.07cm);
 \fill (-0.5,-1) circle (0.07cm);
 \fill(0.5,-1) circle (0.07cm);
 \fill (0,0) circle (0.07cm);
 \draw[->] (0,1) --(0,0) node [draw=none,pos=0,right] {$a$};
 \draw[->] (0,0) --(-0.5,-1) node [draw=none,pos=0,right] {$b$} node [draw=none,pos=1,left] {$c$};
 \draw [->] (0,0) --(0.5,-1) node [draw=none,pos=1,right] {$d$};
 \end{tikzpicture}
}
\hspace{10pt}
  \subfloat[]{
\begin{tikzpicture}[>=latex,scale=0.5]
\tiny
 \tikzstyle{every node} = [circle,draw=black]
 \fill (0,1) circle (0.07cm);
 \fill (-0.5,-1) circle (0.07cm);
 \fill(0.5,-1) circle (0.07cm);
 \draw[->] (0,1) --(-0.5,-1) node [draw=none,pos=0,right] {$a$} node [draw=none,pos=1,left] {$b$};
 \draw [->] (0,1) --(0.5,-1) node [draw=none,pos=1,right] {$c$};
 \end{tikzpicture}
}
\hspace{10pt}
 \subfloat[]{
\begin{tikzpicture}[>=latex,scale=0.5]
\tiny
 \tikzstyle{every node} = [circle,draw=black]
 \fill (-0.5,1) circle (0.07cm);
 \fill (-0.5,-1) circle (0.07cm);
 \fill(0.5,1) circle (0.07cm);
 \fill (0.5,-1) circle (0.07cm);
 \draw[->] (-0.5,1) --(-0.5,-1)node [draw=none,pos=0,right] {$a$} node [draw=none,pos=1,left] {$b$};
 \draw[->] (0.5,1) --(0.5,-1)node [draw=none,pos=0,right] {$c$} node [draw=none,pos=1,left] {$d$};
 \end{tikzpicture}
}
}
   \end{center}
    \caption{An edge in $\compgraph$ and structures of its valid embedding (a) An edge $\gamma$ in $\compgraph$ (b) $\embedding(\gamma) = \{abc,abd\}$ (c) $\embedding(\gamma) = \{ab,ac\}$ (d) $\embedding(\gamma) = \{ab,cd\}$}
    \label{fig:emb_type_valid}
  \end{figure}
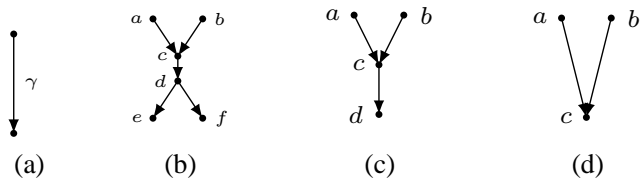
\begin{figure}[h]
 \begin{center}
    \resizebox{0.5\textwidth}{!}{
      \subfloat[]{
\begin{tikzpicture}[>=latex,scale=0.5]
\scriptsize
 \tikzstyle{every node} = [circle,draw=black]
 \tiny
 \fill (0,1) circle (0.07cm);
 \fill (0,-1) circle (0.07cm);
 \draw[->] (0,1) -- (0,-1) node [draw=none,midway,right] {$\gamma$};
 \end{tikzpicture}
}
\hspace{10pt}
\subfloat[]{
\begin{tikzpicture}[>=latex,scale=0.5]
\tiny
 \tikzstyle{every node} = [circle,draw=black]
 \fill (-0.5,1) circle (0.07cm);
 \fill (-0.5,-1) circle (0.07cm);
 \fill(0.5,1) circle (0.07cm);
 \fill (0.5,-1) circle (0.07cm);
 \fill (0,0.25) circle (0.07cm);
 \fill (0,-0.25) circle (0.07cm);
 \draw[->] (-0.5,1) --(0,0.25)node [draw=none,pos=0,left] {$a$} node [draw=none,pos=1,left] {$c$};
 \draw[->] (0.5,1) --(0,0.25)node [draw=none,pos=0,right] {$b$} ;
 \draw[->] (0,0.25) -- (0,-0.25) node [draw=none,pos=1,left] {$d$};
 \draw[->] (0,-0.25) -- (-0.5,-1) node [draw=none,pos=1,left] {$e$};
 \draw[->] (0,-0.25) --(0.5,-1) node [draw=none,pos=1,right] {$f$};
 \end{tikzpicture}
}
\hspace{10pt}
\subfloat[]{
\begin{tikzpicture}[>=latex,scale=0.5]
\scriptsize
 \tikzstyle{every node} = [circle,draw=black]
 \fill (-0.5,1) circle (0.07cm);
 \fill (0.5,1) circle (0.07cm);
 \fill(0,0) circle (0.07cm);
 \fill (0,-1) circle (0.07cm);
 \draw[->] (-0.5,1) --(0,0)node [draw=none,pos=0,left] {$a$} node [draw=none,pos=1,left] {$c$};
 \draw[->] (0.5,1) --(0,0) node [draw=none,pos=0,right] {$b$};
 \draw[->] (0,0) -- (0,-1)node [draw=none,pos=1,left] {$d$};
 \end{tikzpicture}
}
\hspace{10pt}
\subfloat[]{
\begin{tikzpicture}[>=latex,scale=0.5]
\scriptsize
 \tikzstyle{every node} = [circle,draw=black]
 \fill (-0.5,1) circle (0.07cm);
 \fill (0.5,1) circle (0.07cm);
 \fill (0,-1) circle (0.07cm);
 \draw[->] (-0.5,1) --(0,-1) node [draw=none,pos=0,left] {$a$} node [draw=none,pos=1,left] {$c$};
 \draw[->] (0.5,1) -- (0,-1)node [draw=none,pos=0,right] {$b$};
 \end{tikzpicture}
}
   }
   \end{center}
    \caption{An edge in $\compgraph$ and structures of its invalid embedding (a) An edge $\gamma$ in $\compgraph$  (b) $\embedding(\gamma) = \{acde,bcdf\}$ (c) $\embedding(\gamma) =\{acd,bcd\}$ (d) $\embedding(\gamma) = \{ac,bc\}$ }
    \label{fig:emb_type_invalid}
  \end{figure}

\section{Proof of lemmas of Section~\ref{sec:step2}}
\label{app:lemmas}

\subsection{Proof of Lemma~\ref{lm:valid_dag}}

 \begin{enumerate}
  \item Observe that each source vertex of type $S_{ixy}^{*}$ in Fig.~\ref{fig:diamond_gadget}(b) has exactly one outgoing edge of weight $4$ and $\omega_p$ has only incoming edges.
  \item This directly follows from Figs.~\ref{fig:diamond_gadget}(b) and ~\ref{fig:edge_gadget}(b).
  \item First observe that the graph shown in Fig.~\ref{fig:diamond_gadget}(b) has no directed cycles. Moreover the gadget of Fig.~\ref{fig:edge_gadget}(b) does not add any directed cycle as well. This shows that every gadget which replaces an edge $(x,y) \in E_{H}$ is a DAG. Observe that any vertex $x \in V_{H}$ is a part of exactly three such gadgets (each for one of its edges). Thus $x$ has incoming edges from $6$ sources and has outgoing edges to the intermediate vertices inside these gadgets. All the intermediate vertices of a gadget finally go to the sink $\omega_p.$ There are no edges across these gadget thus ensuring that the whole $\compgraph$ is also a DAG.
  \item Every source vertex has exactly one outgoing edge of weight $4$ and every intermediate vertex, i.e., $a_{xy},b_{xy},c_{xy},d_{xy},$ of the gadget has exactly $2$ outgoing edges. Every vertex $x \in V_{H}$ is a part of exactly three gadgets thus has exactly $6$ outgoing edges (two from each gadget). 
  \item All outgoing edges of any source have weight $4.$ Every vertex $x \in V_{H}$ in Fig.~\ref{fig:diamond_gadget}(b) has six outgoing edges of weight one thus after applying the gadget of Fig.~\ref{fig:edge_gadget}(b), it has six outgoing edges of weight $6 \times 1 + 1 = 7.$ Similarly the intermediate vertices have two outgoing edges of weight $2\times 4 +1 =9.$ Thus every edge has bounded weight and the maximum weight of any edge is $9.$
   \end{enumerate}

\subsection{Proof of Lemma~\ref{lm:edgewt}}

 Let $\embedding$ be the minimum cost embedding of $\compgraph$ on $\net$ of cost $C$ in which one (or more) edge of weight $z$ from the gadget Fig.~\ref{fig:edge_gadget}(b) is exposed. In other words, in embedding $\embedding$ some $u_i^{'}$ is mapped to a vertex in $\net$ to which $u$ is not mapped. We modify $\embedding$ by mapping $u_i^{'}$ to the vertex where $u$ is mapped. The modified embedding $\embedding'$ always has cost lesser than the cost of $\embedding$ which contradicts the fact that $\embedding$ is the minimum cost embedding. We explain one such case in detail below.
 \begin{enumerate}
  \item Consider the case when $\embedding(u) = \alpha, \embedding(u_1^{'}) = \embedding(u_2^{'}) = \beta, \embedding(u_1) = \gamma$ and $\embedding(u_2) = \delta.$ In other words, only one of the weight $z$ edge is exposed but both the edges of weight $l_1$ and $l_2$ are exposed. Let $y(\alpha,\beta) = y_1, y(\beta,\gamma) = y_2$ and $y(\beta,\delta) = y_3.$ Then the cost of embedding $\embedding$ coming from this structure is $C = y_1 z + y_2 l_1 + y_3 l_2.$ Now consider the embedding $\embedding'$ where $u_1^{'},u_2^{'}$ are mapped to $\alpha$ keeping all the other vertices at the same location as $\embedding.$ Note that $y(\alpha,\gamma) \leq y_1+y_2$ and $y(\alpha,\delta) \leq y_1+y_3.$ The cost of $\embedding'$ is $C'\leq (y_1+y_2)l_1 + (y_1+y_3)l_2 \leq 2y_1 \max(l_1,l_2) + y_2l_1 + y_3l_2 < y_1z + y_2l_1+ y_3l_2 = C.$ Thus we have an embedding $\embedding'$ where none of the weight $z$ edge is exposed and has cost strictly less than that of $\embedding.$
  \end{enumerate}
The embedding $\embedding'$ and its cost can be computed in the similar manner for other cases of the mappings of various vertices with $C' < C.$ 

\subsection{Proof of Lemma~\ref{lm:mincost_rembedding}}
 A $3$-way multiterminal cut of a graph is the problem of partitioning the vertices into three parts such that each part has exactly one terminal and the weight of the multiterminal cut (defined as the sum of the weights of edges across the parts) is minimized. 
 
 Recall that the network graph $\net$ created in Theorem~\ref{thm:flow-to-cost} is a complete graph on three vertices, namely $S_1,S_2,t,$ with unit edge weights. We create an embedding $\embedding$ of the gadget from a $3$-way cut with weight $W$ of Fig.~\ref{fig:diamond_gadget}(a) as follows: Map the vertices which are with $S_{1xy}$ in the cut to $S_1$ in the embedding. Similarly map a vertex to $S_2$ or $t$ if it is with $S_{2xy}$ or $\omega_p$ in the cut, respectively. Map the intermediate vertices $u_1^{'},\ldots,u_2^{'}$ of Fig.~\ref{fig:edge_gadget}(b) to wherever $u$ is being mapped by the earlier step. It is easy to observe that $\embedding$ is a valid embedding of the gadget.
 
Now we show that the cost of $\embedding$ is $W.$
 Recall that the cost of an embedding is defined by Equation~\eqref{eq:orig_cost} and an edge of the gadget is said to be exposed if its weight is counted in computing the cost of the embedding. In the following arguments we show that an edge of Fig.~\ref{fig:diamond_gadget}(b) is exposed in the embedding iff the corresponding edge of Fig.~\ref{fig:diamond_gadget}(a) is in the cut.

\begin{enumerate}
 \item Consider an edge $(S_{1xy},*)$ of Fig.~\ref{fig:diamond_gadget}(a). If it is in the cut then its end points, i.e., $S_{1xy}$ and $*,$ are in two separate partitions. This in turn implies that the vertex $*$ of Fig.~\ref{fig:diamond_gadget}(b) is not mapped to $S_1$ in embedding $\embedding$ and the edge $(S_{1xy}^{*},*)$ is exposed in $\embedding.$ Similarly, if an edge $(S_{2xy},*)$ of Fig.~\ref{fig:diamond_gadget}(a) is in the cut then the corresponding edge $(S_{2xy}^{*},*)$ of Fig.~\ref{fig:diamond_gadget}(b) is exposed in $\embedding.$ Note that weights of $(S_{ixy},*)$ (Fig.~\ref{fig:diamond_gadget}(a)) and $(S_{ixy}^{*},*)$ (Fig.~\ref{fig:diamond_gadget}(b)) for $i \in \{1,2\}$ are same thus contributing to the same weight in the cut as well as the cost of $\embedding.$
 
 \item Now consider the edges $(x,a_{xy})$ and $(x,c_{xy})$ of Fig.~\ref{fig:diamond_gadget}(a). If both the edges are in the cut then there are two possibilities: either $x,a_{xy},c_{xy}$ all are in separate partitions or $x$ is in one partition but $a_{xy},c_{xy}$ are together in different partition. Observe the corresponding edges in Fig.~\ref{fig:diamond_gadget}(b). They are replaced by the structure of Fig.~\ref{fig:edge_gadget}(b) with $a_{xy}^{'},c_{xy}^{'}$ as intermediate vertices between $x$ and $a_{xy},c_{xy}$ respectively. Note that under embedding $\embedding,$ vertices $a_{xy}^{'},c_{xy}^{'}$ are mapped wherever $x$ is mapped and $a_{xy},c_{xy}$ are mapped to either different or same vertices (depending on them being in different or same partitions in the cut). In either case the edges $(a_{xy}^{'},a_{xy})$ and $(c_{xy}^{'},c_{xy})$ are exposed in the embedding if $(x,a_{xy})$ and $(x,c_{xy})$ of Fig.~\ref{fig:diamond_gadget}(a) are in the cut thus contributing to the same weight in $\embedding$'s cost. Same argument holds for all the outgoing edges from vertices $x,y,a_{xy},b_{xy},c_{xy},d_{xy}$ of Fig.~\ref{fig:diamond_gadget}(b).
 
 \item Finally note that an edge of Fig.~\ref{fig:diamond_gadget}(b) is exposed only if its end points are mapped to different vertices in $\embedding$ which in turn implies that the corresponding edge of Fig.~\ref{fig:diamond_gadget}(a) is in cut. The weight $z$ edges of Fig.~\ref{fig:edge_gadget}(b) are never exposed in $\embedding$ as their endpoints are always mapped to same vertex in $\embedding.$ 
 
\end{enumerate}
 This proves that the cost of $\embedding$ is indeed $W$ which is same as the weight of the $3$-way cut.

\subsection{Proof of Lemma~\ref{lm:cost_cut}}

Recall that for every edge $(x,y) \in E_{H}$ there is a gadget of Fig.~\ref{fig:diamond_gadget}(b) (along with Fig.~\ref{fig:edge_gadget}(b)) in $\compgraph$ and the network graph $\net$ has only three vertices. Given an embedding $\embedding$ with multiple mappings for a vertex we construct the embedding $\embedding'$ with single mapping in the following steps. 

\begin{enumerate}
 \item If any intermediate vertex of Fig.~\ref{fig:edge_gadget}(b), i.e., $u_1^{'},\ldots,u_h^{'},$ is mapped to multiple vertices then in $\embedding'$ map all its copies to wherever $u$ is mapped in $\embedding$ keeping the rest of the vertices at the same place. This will only reduce the cost of the resulting embedding.
 \item Observe that the vertices $b_{xy},c_{xy}$ of Fig.~\ref{fig:diamond_gadget}(b) have only one outgoing edge which is going to $\omega_p.$ As the mapping of $\omega_p$ is fixed to $t \in \netnodes$  in any valid embedding, the outputs of $b_{xy},c_{xy}$ are required only at one vertex in the embedding. Thus, the operations performed at these nodes cannot be performed at multiple vertices in the network graph and $b_{xy},c_{xy}$ are not mapped to multiple vertices in any valid embedding. 
 \item Consider the vertex $a_{xy}$ \footnote{$a_{xy}$ has outgoing edges to $\omega_p,b_{xy}$ and both are mapped to only one vertex under a valid embedding.} and let it be mapped to two vertices in $\netnodes$ under embedding $\embedding.$ There are three possible mappings of $a_{xy}$ in this case and we show that in each case mapping it to only one of the vertices brings down the cost of the embedding. 
 \begin{enumerate}
  \item Let $a_{xy}$ be mapped to $S_2$ and $t$ under embedding $\embedding.$ Create an embedding $\embedding'$ where $a_{xy}$ is mapped to only $t$ keeping the mapping of all the vertices same as that of $\embedding.$ Then, $C(\embedding') < C(\embedding) - w(S_{1xy}^{a},a_{xy}) y(S_1,S_2) + w(a_{xy},b_{xy}) y(S_2,t) = C(\embedding) - 4+1 < C(\embedding).$
  \item Let $a_{xy}$ be mapped to $S_1$ and $t$ under $\embedding.$ Create the embedding $\embedding'$ where $a_{xy}$ is mapped only to $S_1$ keeping the mapping of all the vertices same as that of $\embedding.$ Then, $C(\embedding') = C(\embedding) - w(S_{1xy}^{a},a_{xy}) y(S_1,t) + w(a_{xy},\omega_p) y(S_1,t) = C(\embedding) -4+4.$
  \item Let $a_{xy}$ be mapped to $S_1$ and $S_2$ under $\embedding.$ Create the embedding $\embedding'$ where $a_{xy}$ is mapped only to $S_1$ keeping the mapping of all the vertices same as that of $\embedding.$ It is easy to observe that $C(\embedding') \leq C(\embedding) -3$ in this case. 
 \end{enumerate}
 The vertex $d_{xy}$ can also be mapped only to one vertex by similar arguments.
\item Now consider the vertex $x$ in the $(x,y)$ gadget. Since $x$ has two outgoing neighbors in this gadget (namely $a_{xy},c_{xy}$) and each of them can be mapped to only one vertex, $x$ in turn can be mapped to at most two vertices for this gadget. We create the embedding $\embedding'$ of reduced cost as follows.
 \begin{enumerate}
  \item Let $x$ be mapped to $S_1$ and $S_2$ under $\embedding$ for this gadget. Then create the embedding $\embedding'$ where $x$ is mapped only to $S_1$ keeping the mapping of all the vertices same as that of $\embedding.$ Then, $C(\embedding') = C(\embedding) -  w(S_{1xy}^{x},x) y(S_1,S_2) + w(x,a_{xy})y(S_1,S_2) = C(\embedding) -4+1 < C(\embedding).$
  \item Let $x$ be mapped to $S_1$ and $t$ under $\embedding.$ Create $\embedding'$ where $x$ is mapped to $S_1$ keeping the mapping of all the vertices same as that of $\embedding.$ It is easy to observe that $C(\embedding') \leq C(\embedding) -4 -4 +2 < C(\embedding).$ Similarly if $x$ is mapped to $S_2$ and $t$ then get new embedding by mapping it to $S_2.$
 \end{enumerate}
 In this way for any edge $(x,y)$ each vertex of the corresponding gadget can be mapped to only one vertex in $\embedding'$ and $C(\embedding') \leq C(\embedding).$
 \item Recall that every $x \in V_H$ has three edges in $H,$ thus $x$ is a part of three gadgets. Till now we have made sure that individually for each gadget $x$ is mapped to only one vertex of $\net$ but it is possible that it is mapped to more than one vertex across the gadgets. Let $(x,y)$ and $(x,z)$ be two edges for whose gadgets $x$ is mapped to separate vertices in $\embedding.$ Let $x$ be mapped to $S_1$ for $(x,y)$ gadget and to $S_2$ for $(x,z)$ gadget. Create the embedding $\embedding'$ where $x$ is mapped to $S_1$ for $(x,z)$ gadget keeping the mapping of all the other vertices same as that of $\embedding.$ Observe that in embedding $\embedding$ to compute $x$ at $S_1$ edges $(S_{2xz}^x,x),(S_{2xy}^x,x)$ and to compute it at $S_2$ edges $(S_{1xz}^x,x),(S_{1xy}^x,x)$ are exposed. While in $\embedding'$ as $x$ is computed only at $S_1$ the edges $(S_{1xz}^x,x),(S_{1xy}^x,x)$ will not be exposed thus reducing the cost of embedding by $8.$ At the same time, at most the outgoing edges of $x$ from $(x,z)$ gadget, i.e., $(x,a_{xz})(y,c_{xz}),$ might get exposed. Thus $C(\embedding') < C(\embedding) -8+2.$
\end{enumerate}

In this way we get an embedding $\embedding'$ in which each vertex of $\compgraph$ is mapped to only one vertex of $\net$ and has cost at most that of $\embedding.$

\end{document}